\documentclass{llncs}

\usepackage[cmex10]{amsmath}
\usepackage{array}
\usepackage{booktabs}
\usepackage{hyperref}
\usepackage{microtype}
\usepackage{multirow}
\usepackage{paralist}
\usepackage{rotating}
\usepackage{subfig}
\usepackage{tabularx}
\usepackage{temporal}
\usepackage{url}
\usepackage{xspace}
\newcolumntype{Z}{>{\centering\let\newline\\\arraybackslash\hspace{0pt}}X}

\usepackage{tikz}
\usetikzlibrary{arrows,automata,backgrounds,calc,positioning}

\usepackage{tikz-timing}
\usetikztiminglibrary[rising arrows]{clockarrows}
\usepackage{xparse}
\NewDocumentCommand{\busref}{som}{\texttt{#3\IfValueTF{#2}{[#2]}{}\IfBooleanTF{#1}{\#}{}}}

\usepackage{algorithm}
\usepackage[noend]{algpseudocode}
\algrenewcommand\algorithmicindent{1em}
\algdef{SE}[PROCEDURE]{ProcedureRet}{EndProcedureRet}[3]
  {\algorithmicprocedure\ \textproc{#1}\ifthenelse{\equal{#2}{}}{}{$\left(#2\right)$}, \textbf{returns}: {#3}}
  {\algorithmicend\ \algorithmicprocedure}

\newcommand{\refalg}[1]{{Algorithm~\ref{alg:#1}\xspace}}
\newcommand{\refeq}[1]{{Equation~\ref{eq:#1}\xspace}}
\newcommand{\refex}[1]{{Example~\ref{ex:#1}\xspace}}
\newcommand{\reffig}[1]{{Fig.~\ref{fig:#1}\xspace}}
\newcommand{\reflem}[1]{{Lemma~\ref{th:#1}\xspace}}
\newcommand{\reftab}[1]{{Table~\ref{tab:#1}\xspace}}

\newcommand{\eucropis}{Eu:CROPIS\xspace}
\newcommand{\partystrategy}{\textsf{PARTYStrategy}\xspace}
\newcommand{\party}{\textsf{PARTY}\xspace}

\newcommand{\alp}{\Sigma}
\newcommand{\concat}{\circ}
\newcommand{\constr}{\Theta} 
\newcommand{\faultform}{\delta}
\newcommand{\faultfreqset}{\mathsf{Frq}}
\newcommand{\faultfreq}{\mathsf{frq}}
\newcommand{\faultkind}{\kappa}
\newcommand{\inalp}{\Sigma_I}
\newcommand{\inp}{I}
\newcommand{\intrace}{{\overline{\sigma_I}}}
\newcommand{\letter}{\sigma}
\newcommand{\mealyset}{\mathsf{Mealy}}
\newcommand{\modelcheck}{\mathsf{modelcheck}}
\newcommand{\mooreset}{\mathsf{Moore}}
\newcommand{\moore}{\mathcal{M}}
\newcommand{\outalp}{\Sigma_O}
\newcommand{\outp}{O}
\newcommand{\outtrace}{{\overline{\sigma_O}}}
\newcommand{\real}{\!\mathrel {||}\joinrel \Relbar\!}
\newcommand{\strat}{\mathcal{T}}
\newcommand{\suite}{\text{TS}}
\newcommand{\syntp}{\mathsf{synt}_p}
\newcommand{\synt}{\mathsf{synt}}
\newcommand{\sys}{\mathcal{S}}
\newcommand{\trace}{{\overline{\sigma}}}
\newcommand{\unreal}{\textsf{un\-real\-izable}\xspace}

\newcommand{\car}{\mathsf{c}}
\newcommand{\h}{\mathsf{h}}
\newcommand{\f}{\mathsf{f}}
\newcommand{\p}{\mathsf{p}}

\newcommand{\nommode}{\text{\tt mode$_{\text{\tt 1}}$}\xspace}
\newcommand{\redmode}{\text{\tt mode$_{\text{\tt 2}}$}\xspace}
\newcommand{\normerr}{\text{\tt err$_{\text{\tt nc}}$}\xspace}
\newcommand{\criterr}{\text{\tt err$_{\text{\tt s}}$}\xspace}
\newcommand{\reset}{\text{\tt reset}\xspace}
\newcommand{\safemode}{\text{\tt safemode}\xspace}
\newcommand{\redon}{\text{\tt on$_{\text{\tt 2}}$}\xspace}
\newcommand{\nomon}{\text{\tt on$_{\text{\tt 1}}$}\xspace}
\newcommand{\redoff}{\text{\tt off$_{\text{\tt 2}}$}\xspace}
\newcommand{\nomoff}{\text{\tt off$_{\text{\tt 1}}$}\xspace}
\newcommand{\lastupisnom}{\text{\tt lastup}\xspace}
\newcommand{\allowswitch}{\text{\tt allowswitch}\xspace}

\begin{document}

\title{Synthesizing Adaptive Test Strategies from Temporal Logic Specifications}
\author{Roderick Bloem$^{1}$ \and Goerschwin Fey$^{2,3}$ \and Fabian Greif$^{3}$ \and Robert K\"{o}nighofer$^{1}$ \and\\%
        Ingo Pill$^{1}$ \and Heinz Riener$^{3,4}$ \and Franz R\"{o}ck$^{1}$}
\institute{%
  $^1$Graz University of Technology, Graz, Austria \\
  $^2$Hamburg University of Technology, Hamburg, Germany \\
  $^3$German Aerospace Center, Bremen, Germany \\
  $^4$EPFL, Lausanne, Switzerland
}

\maketitle

\begin{abstract}
Constructing good test cases is difficult and time-consuming, especially if the system under test is still under development and its exact behavior is not yet fixed.  We propose a new approach to compute  test strategies for reactive systems from a given temporal logic specification using formal methods.  The computed strategies are guaranteed to reveal certain simple faults  in \emph{every} realization of the specification and for \emph{every} behavior of the uncontrollable part of the system's environment.  The proposed approach supports different assumptions on occurrences of faults (ranging from a single transient fault to a persistent fault) and by default aims at unveiling the weakest one.  Based on well-established hypotheses from fault-based testing, we argue that such tests are also sensitive for more complex bugs.  Since the specification may not define the system behavior completely, we use reactive synthesis algorithms with partial information. The computed strategies are \emph{adaptive test strategies} that react to behavior at runtime.  We work out the underlying theory of adaptive test strategy synthesis and present experiments for a safety-critical component of a real-world satellite system.  We demonstrate that our approach can be applied to industrial specifications and that the synthesized test strategies are capable of detecting bugs that are hard to detect with random testing.
\end{abstract}

\section{Introduction}

Model checking~\cite{ClarkeE81,QueilleS82} is an algorithmic approach to prove that a model of a system adheres to its specification.  However, model checking cannot always be applied effectively to obtain confidence in the correctness of a system.  Possible reasons include scalability issues, third-party IP components for which no code or detailed model is available, or a high effort for building system models that are sufficiently precise. Moreover, model checking cannot verify the final and ``live'' product but only an (abstracted) model.

Testing is a natural alternative to complement formal methods like model checking, and automatic test case generation helps keeping the effort acceptable.  Black-box testing techniques, where tests are derived from a specification rather than the implementation, are particularly attractive: first, tests can be computed before the implementation phase starts, and thus guide the development.  Second, the same tests can be reused across different realizations of a given specification.  Third, a specification is usually much simpler than its implementation, which gives a scalability advantage.  At the same time, the specification focuses on critical functional aspects that require thorough testing.  Fault-based techniques~\cite{JiaH11} are particularly appealing, where the computed tests are guaranteed to reveal all faults in a certain fault class --- after all, the foremost goal in testing is to detect bugs.

Methods to derive tests from declarative requirements (see, e.g.,~\cite{FraserWA09}) are sparse.  One issue in this setting is controllability: the requirements leave plenty of implementation freedom, so they cannot be used to fully predict the system behavior for all given inputs.  Consequently, test cases have to be \emph{adaptive}, i.e., able to react to observed behavior at runtime, rather than being fixed input sequences.  This is particularly true for \emph{reactive systems} that continuously interact with their environment.  Existing methods often work around this complication by requiring a deterministic system model as additional input~\cite{FraserWA09b}.  Even a probabilistic model fixes the behavior in a way not necessarily required by the specification.

In previous work, we presented a fault-based approach to compute adaptive test strategies for reactive systems~\cite{BloemKPR16}.  This approach generates tests that enforce certain coverage goals for \emph{every} implementation of a provided specification.  The generated tests can be used across realizations of the specification that differ not only in implementation details but also in their observable behavior.  This is, e.g., useful for standards and protocols that are implemented by multiple vendors or for systems under development, where the exact behavior is not yet fixed.

\reffig{testing setup} outlines the assumed testing setup and shows how the approach for synthesizing adaptive test strategies (illustrated in black) can be integrated in an existing testing flow.  The user provides a specification $\varphi$, which describes the requirements of the system under test (SUT) and additionally a fault model~$\faultform$, which defines the coverage goal in terms of a class of faults for which the tests shall cause a specification violation.  Both the specification and the coverage goal are expressed in Linear Temporal Logic (LTL)~\cite{Pnueli77}.  By default, our approach supports the detection of transient and permanent faults and distinguishes four fault occurrence frequencies: faults that occur at least (1) once, (2) repeatedly, (3) from some point on, or (4) permanently.  The approach then automatically synthesizes a test strategy to reveal a fault for the lowest frequency possible.  Such a test strategy guarantees to cause a specification violation if the fault occurs with the defined fault occurrence (and all higher fault occurrence frequencies) and the test is executed long enough.  Besides the four default fault occurrence frequencies, a user can also provide a custom frequency using LTL.

Under the hood, reactive synthesis~\cite{PnueliR89} with partial information~\cite{KupfermanV00} is used, which provides strong guarantees about all uncertainties: if synthesis is successful and if the computed tests are executed long enough, they reveal all faults from the fault model for every realization of the specification and every behavior of the uncontrollable part of the system's environment.  Uncontrollable environment aspects can be seen as part of the system for the purpose of testing.  Finally, existing techniques from runtime verification~\cite{BauerLS11} can be used to build an oracle that checks the system behavior against the specification while tests are executed.\footnote{While the semantics of LTL are defined over infinite execution traces, we can only run the tests for a finite amount of time.  This can result in inconclusive verdicts~\cite{BauerLS11}.  We exclude this issue from the scope of this paper, relying on the user to judge when tests have been executed long enough, and on existing research on interpreting LTL over finite traces~\cite{MorgensternGS12,HavelundR01,GiacomoV13,GiacomoMM14}.}

\begin{figure}[t]
  \centering
  \includegraphics[width=1.0\textwidth]{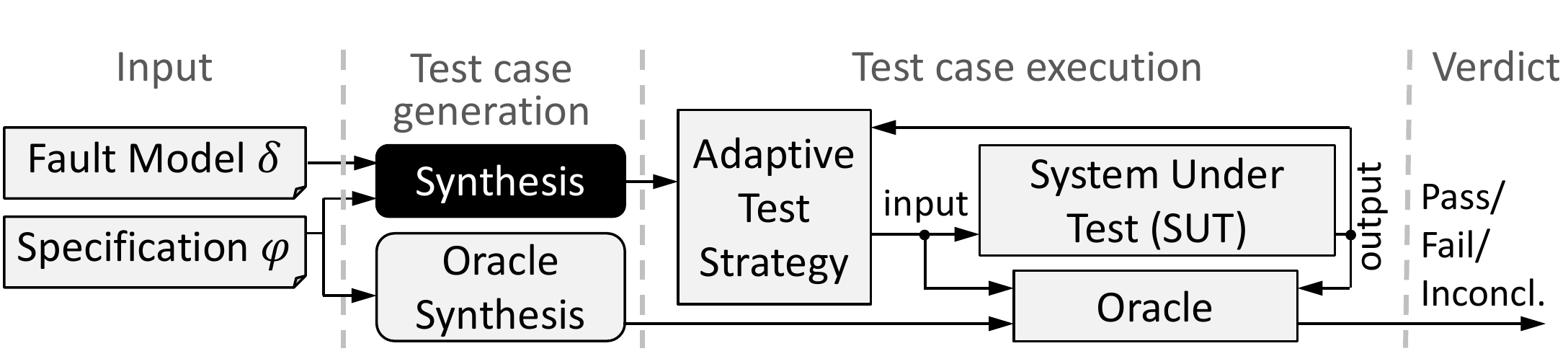}
  \caption{Testing setup: this paper focuses on test strategy synthesis.}
  \label{fig:testing setup}
\end{figure}

This paper is an extension of~\cite{BloemKPR16}.  In summary, this paper presents the following contributions:
\begin{itemize}
\item An approach to compute adaptive test strategies for reactive systems from temporal specifications that provide implementation freedom.  The tests are guaranteed to reveal certain bugs for \emph{every} realization of the specification.
\item The underlying theory is considered in detail, i.e.,  we show that the approach is sound and complete for many interesting cases and provide additional solutions for other cases that may arise in practice.
\item A proof of concept tool, called \partystrategy\footnote{\partystrategy, https://www.iaik.tugraz.at/content/research/scos/tools/},  that is capable of generating multiple different test strategies, implemented on top of the synthesis tool \party~\cite{KhalimovJB13}.
\item A post-processing procedure to generalize a test strategy by eliminating input constraints not necessary to guarantee a coverage goal.
\item A case study with a safety-critical software component of a real-world satellite system developed in the German Aerospace Center (DLR).  We specify the system in LTL, synthesize test strategies, and evaluate the generated adaptive test strategies using code coverage and mutation coverage metrics. Our synthesized test strategies increase both the mutation coverage as well as the code coverage of random test cases  by activating behaviors that require complex input sequences that  are unlikely to be produced by random testing.
\end{itemize}

The remainder of this paper is organized as follows: Section~\ref{sec:ex} illustrates our approach and presents a motivating example.  Section~\ref{sec:rel_work} discusses related work.  Section~\ref{sec:prelim} gives preliminaries and notation.  Our test case generation approach is then worked out in detail in Section~\ref{sec:synt}.  Section~\ref{sec:experimentals} presents the case study and discusses results.  Section~\ref{sec:concl} concludes.

\section{Motivating Example}
\label{sec:ex}

\begin{figure}[t]
  \centering
  \includegraphics[width=.5\textwidth]{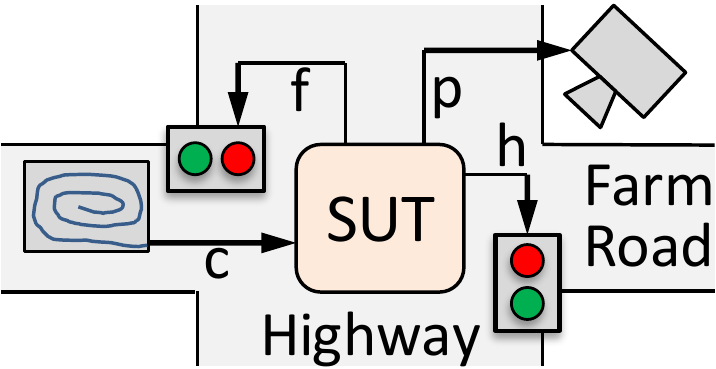}
  \caption{Traffic light example.}
  \label{fig:traffic lights}
\end{figure}

Let us develop a traffic light controller for the scenario depicted in  \reffig{traffic lights}. For this highway and farmroad crossing, the controller's  Boolean input signal $c$ describes whether a car is idling at the farmroad.  Boolean outputs $h$ and $f$ control the highway and farmroad traffic lights  respectively, where a value of $\true$ means a green light. Output $p$ controls  a camera that takes a picture if a car on the farmroad makes a fast start, i.e.,  races off immediately when the farmroad light turns green. The controller then should implement the following critical properties:
\begin{enumerate}
\item The traffic lights must never be green simultaneously.\label{item:p1}
\item If a car is waiting at the farmroad, $f$ eventually turns $\true$.\label{item:p2}
\item If no car is waiting at the farmroad, $h$ eventually becomes $\true$.\label{item:p3}
\item A picture is taken if a car on the farmroad makes a fast start.\label{item:p4}
\end{enumerate}

\bigskip

We model the four properties in Linear Temporal Logic (LTL)~\cite{Pnueli77} as
\begin{align}
\varphi_1 &= \always(\neg \f \vee \neg \h) \\
\varphi_2 &=\always(\car \rightarrow \eventually \f) \\
\varphi_3 &=\always(\neg \car \rightarrow \eventually \h) \\
\varphi_4 &=\always\bigl((\neg \f \wedge \nextt(\car \wedge \f \wedge \nextt \neg \car)) \leftrightarrow \nextt \nextt \p\bigr)
\end{align}
where the operator $\always$ denotes \emph{always}, $\eventually$ denotes \emph{eventually}, and $\nextt$ denotes \emph{in the nextstep}.

The resulting specification is then:
\begin{align*}
\varphi = \varphi_1 \land \varphi_2 \land \varphi_3 \land \varphi_4
\end{align*}

\vspace{1mm}
To compute a test strategy (only from the specification) that enforces a specification violation by the system under the existence of a certain fault (or class of faults), we have some requirements for our approach.

\paragraph{Enforcing test objectives}
To mitigate scalability issues, we compute test cases directly from the specification $\varphi$.  Note that $\varphi$ focuses on the desired properties only, and allows for plenty of implementation freedom.  Our goal is to compute tests that \emph{enforce} certain coverage objectives \emph{independent} of this implementation freedom.  Some uncertainties about the SUT behavior may actually be rooted in uncontrollable environment aspects (such as weather conditions) rather than implementation freedom inside the system.  But for our testing approach, this makes no difference.  We can force the farmroad's traffic light to turn green ($\f$=$\true$) by relying on a correct implementation of Property~\ref{item:p2} and setting $\car$=$\true$. Depending on how the system is implemented, $\f$=$\true$ might also be achieved by setting $\car$=$\false$ all the time, but this is not guaranteed.

\begin{figure}[t]
  \centering
  \subfloat{\begin{tikzpicture}[>=latex,->,auto,initial text={},initial distance=3mm]
\node[state,initial,inner sep=0]  at  (0,0)       (S0) {$\neg \car$};
\node[state,inner sep=0]          at  (1.5,0)     (S1) {$\car$};
\node[state,inner sep=0]          at  (3.0,0)     (S2) {$\neg \car$};

\path
(S0) edge [loop above]  node[xshift=-3mm,yshift=-4mm] {$\f$} (S0)
(S0) edge [bend left]  node[xshift=0mm,yshift=-5mm] {$\neg \f$} (S1)
(S1) edge [loop above]  node[xshift=-4mm,yshift=-4mm] {$\neg \f$} (S1)
(S1) edge [bend left]  node[xshift=0mm,yshift=-5mm] {$\f$} (S2)
(S2) edge [loop above]  node[xshift=-5mm,yshift=-4mm] {$\true$} (S2)
;
\end{tikzpicture}}%
  \qquad\qquad
  \subfloat{\begin{tikzpicture}[>=latex,->,auto,initial text={},initial distance=3mm]
\node[state,initial,inner sep=0]  at  (0,0)       (S0) {$\neg \car$};
\node[state,inner sep=0]          at  (1.5,0)     (S1) {$\car$};

\path
(S0) edge [loop above]  node[xshift=-3mm,yshift=-3mm] {$\f$} (S0)
(S1) edge [loop above]  node[xshift=-3mm,yshift=-3mm] {$\neg \f$} (S1)
(S0) edge [bend left]   node[xshift=0mm,yshift=-0.5mm] {$\neg \f$} (S1)
(S1) edge [bend left,line width=1.2pt] node[xshift=0mm,yshift=4.5mm] {$\f$} (S0)
;
\end{tikzpicture}}%
  \caption{Two adaptive test strategies for the traffic light controller: on the left, $\strat_1$ that enforces $\p$ = $\true$ once.  On the right, $\strat_2$ that enforces $\p$ = $\true$ infinitely often.}%
  \label{fig:motivating example1}%
\end{figure}
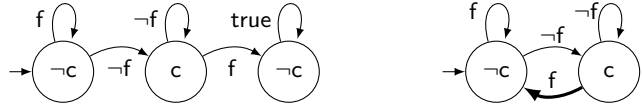

\paragraph{Adaptive test strategies}
Certain test goals may not be enforceable with a static input sequence.  For our example, for $\p$ to be $\true$, a car must do a fast start.  Yet, the specification does not prescribe the exact point in time when the traffic light turns to green.  We thus synthesize \emph{adaptive} test strategies that guide the controller's inputs based on the previous inputs and outputs and, therefore, can take advantage of situational possibilities by exploiting previous system behavior.

\reffig{motivating example1} shows a test strategy $\strat_1$ (on the left) to reach $\p$=$\true$, illustrated as a state machine.  States are labeled by the value of controller \emph{input}~$\car$ (which is an  \emph{output of the test strategy $\strat_1$}).  Edges represent  transitions and are labeled with  conditions on observed output values (since the SUT's outputs are inputs for the  test strategy).  First, $\car$ is set to $\false$ to provoke  $\h$=$\true$ via Property~\ref{item:p3}, implying $\f$=$\false$ via Property~\ref{item:p1}.  As soon as this happens, the strategy traverses to the  middle state, setting $\car$=$\true$ in order to have $\f$=$\true$ eventually (Property~\ref{item:p2}).  As soon as $\f$ switches from $\false$ to $\true$, $\strat_1$ sets $\car$=$\false$ in the rightmost state to trigger a picture (Property~\ref{item:p4}). A system with a permanent stuck-at-0 fault at signal $\p$ is unable to satisfy the specification and the resulting violation can be detected by a runtime verification technique.

\paragraph{Coverage objectives}
We follow a fault-centered approach to define the test objectives to enforce.  The user defines a class of (potentially transient) faults.  Our approach then computes adaptive test strategies (in form of state machines) that detect these faults.  For a permanent stuck-at-$0$ fault at signal~$\p$, our approach could produce the test strategy $\strat_1$ from the previous paragraph: for any correct implementation of $\varphi$, the strategy enforces~$\p$ becoming $\true$ at least once.  Thus, a faulty version where $\p$ is always $\false$ necessarily violates the specification, which can be detected~\cite{BauerLS11} during test strategy execution.  The test strategy $\strat_2$, as shown on the right of \reffig{motivating example1}, is even more powerful since it also reveals stuck-at-$0$ faults for $\p$ that occur not always but only from some point in time onwards.  The difference to $\strat_1$ is mainly in the bold transition, which makes $\strat_2$ enforce $\p$=$\true$ infinitely often rather than only once. Our approach distinguishes four fault occurrence frequencies (a fault occurs at least once, infinitely often, from some point on, or always) and synthesizes test strategies for the lowest one for which this is possible.

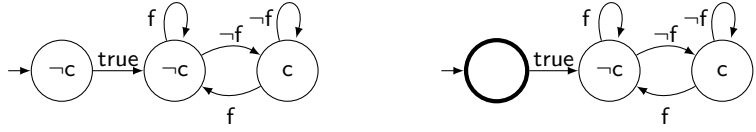
\begin{figure}[t]
  \centering
  \subfloat{\begin{tikzpicture}[>=latex,->,auto,initial text={},initial distance=3mm,node distance=1.5cm]
\node[state,initial,inner sep=0]     at  (0,0) (S0) {$\neg \car$};
\node[state,inner sep=0,right of=S0]           (S1) {$\neg \car$};
\node[state,inner sep=0,right of=S1]           (S2) {$\car$};

\path
(S1) edge [loop above]  node[xshift=-3mm,  yshift=-4mm] {$\f$}      (S1)
(S2) edge [loop above]  node[xshift=-3.5mm,yshift=-4mm] {$\neg \f$} (S2)
(S0) edge []            node[xshift=0mm,yshift=-0.5mm]  {$\true$}       (S1)
(S1) edge [bend left]   node[xshift=0mm,yshift=-0.5mm]  {$\neg \f$} (S2)
(S2) edge [bend left]   node[xshift=0mm,yshift=-0.5mm]  {$\f$}      (S1)
;
\end{tikzpicture}}%
  \qquad\qquad
  \subfloat{\begin{tikzpicture}[>=latex,->,auto,initial text={},initial distance=3mm,node distance=1.5cm]
\node[ultra thick,state,initial,inner sep=0]  at (0,0) (S0) {};
\node[state,inner sep=0,right of=S0]                   (S1) {$\neg \car$};
\node[state,inner sep=0,right of=S1]                   (S2) {$\car$};

\path
(S1) edge [loop above] node[xshift=-3mm,yshift=-4mm]    {$\f$}      (S1)
(S2) edge [loop above] node[xshift=-3.5mm,yshift=-4mm]  {$\neg \f$} (S2)
(S0) edge []           node[xshift=0mm,yshift=-0.5mm]   {$\true$}       (S1)
(S1) edge [bend left]  node[xshift=0mm,yshift=-0.5mm]   {$\neg \f$} (S2)
(S2) edge [bend left]  node[xshift=0mm,yshift=-0.5mm]   {$\f$}      (S1)
;
\end{tikzpicture}}%
  \caption{Two more adaptive test strategies for the traffic light controller: on the left, $\strat_3$ that enforces $\p$=$\true$ infinitely often starting from the second time step.  On the right, $\strat_4$ that generalizes $\strat_3$ by allowing an arbitrary choice for the input in the first time step.}%
  \label{fig:motivating example2}%
\end{figure}

\paragraph{Multiple strategies}
The previously discussed strategies, $\strat_1$ and $\strat_2$, reveal a stuck-at-$0$ fault that manifests permanently at signal $\p$ or a stuck-at-$0$ fault that manifests from some point in time on permanently at signal $\p$, respectively.  Let us now assume that a stuck-at-$0$ fault occurs from some point in time only if a certain input-output interaction happened first, e.g., if $\car$ is $\false$ at the second time step.  Strategy~$\strat_3$ as shown on the left of \reffig{motivating example2} sets $\car$=$\false$ in the second time step.  The output produced by the SUT as a response is not relevant.  The strategy then follows $\strat_2$ to enforce $\p$=$\true$ infinitely often, as before. The two test strategies, $\strat_2$ and $\strat_3$, enforce the same test objective; however when executed they produce different traces.  We argue that considering multiple test strategies for a test objective is necessary to uncover faults in different system implementations and extend our approach to compute a bounded number of test strategies for a given test objective to improve the overall fault coverage while keeping the computational overhead controllable by the user.

\paragraph{Strategy generalization}
The assignment $\car=\false$ in the initial state of $\strat_3$ is neither necessary to activate the fault in the envisioned scenario nor to enforce $\p$=$\true$ infinitely often.  From a testing perspective, the tester is free to make an arbitrary choice for the input to the SUT in the initial state.  As a generalization mechanism of the test strategies, we identify and remove state machine labels not necessary to enforce the test objective.  Strategy~$\strat_4$, illustrated on the right of \reffig{motivating example2}, is similar to $\strat_3$, but differs by only having assignments for input variables in states where the concrete values are necessary to enforce the desired behavior.

\section{Background and Related Work}
\label{sec:rel_work}

\paragraph{Fault-based testing}
Fault-based test case generation methods that use the concept of mutation testing~\cite{JiaH11} seed simple faults into a system implementation (or model) and compute tests that uncover these faults.  Two hypotheses support the value of such tests.  The Competent Programmer Hypothesis~\cite{DeMilloLS78,AcreeBDLS79} states that implementations are mostly close to correct.  The Coupling Effect~\cite{DeMilloLS78,Offutt92} states that tests that detect simple faults are also sensitive to more complex faults.  Our approach also relies on these hypotheses.  However, in contrast to most existing work that considers permanent faults and deterministic system descriptions that define behavior unambiguously, our approach can deal with transient faults and focuses on uncovering faults in \emph{every} implementation of a given LTL~\cite{Pnueli77} specification (and all behaviors of the uncontrollable part of the system's environment).

\paragraph{Adaptive tests}
If the behavior of the system or the uncontrollable part of the environment is not fully specified, tests may have to react to observed behavior at runtime to achieve their goals.  Such adaptive tests have been studied by Hierons~\cite{Hierons06} from a theoretical perspective, relying on fairness assumptions (every non-deterministic behavior is exhibited when trying often enough) or probabilities. Petrenko et al.~compute adaptive tests for trace inclusion \cite{PetrenkoS15,PetrenkoY14,PetrenkoY05} or equivalence~\cite{PetrenkoSY12,LuoBP94,PetrenkoY05} from a specification given as non-deterministic finite state machine, also relying on fairness assumptions.  Our work makes no such assumptions but considers the SUT to be fully antagonistic.  Aichernig et al.~\cite{AichernigBJKST15} present a method to compute adaptive tests from (non-deterministic) UML state machines.  Starting from an initial state, a trace to a goal state, the state that shall be covered by the resulting test case, is searched for every possible system behavior, issuing inconclusive verdicts only if the goal state is not reachable any more.  Our approach uses reactive synthesis to enforce reaching the testing goal for all implementations if this is possible.

\paragraph{Testing as a game}
Yannakakis~\cite{Yannakakis04} points out that  testing reactive systems can be seen as a game between two players: the tester  providing inputs and trying to reveal faults, and the SUT providing outputs and  trying to hide faults.  The tester can only observe outputs and has thus partial  information about the SUT.  The goal is to find a strategy for the tester that  wins against every SUT.  The underlying complexities are studied by Alur et al.~\cite{AlurCY95} in detail.  Our work builds upon reactive synthesis~\cite{PnueliR89}  (with partial information~\cite{KupfermanV00}), which can also be seen as a game.   However, we go far beyond the basic idea.  We combine the game concept with  user-defined fault models, work out the underlying theory, optimize the faults  sensitivity in the temporal domain, and present a realization and experiments  for LTL~\cite{Pnueli77}.  Nachmanson et al.~\cite{NachmansonVSTG04} synthesize  game strategies as tests for non-deterministic software models, but  their approach is not fault-based and focuses on simple reachability goals. A variant of their approach considers the SUT to behave probabilistically with known  probabilities~\cite{NachmansonVSTG04}.  The same model is also used  in~\cite{BlassGNV05}.  Test strategies for reachability goals are also considered by David et al.~\cite{DavidLLN08} for timed automata.

\paragraph{Vacuity detection}
Several approaches~\cite{BeerBER01,KupfermanV03,ArmoniFFGPTV03} aim at finding cases where a temporal specification is trivially satisfied (e.g., because the left side of an implication is false).  Good tests avoid such vacuities to challenge the SUT.  The method by Beer et al.~\cite{BeerBER01} can produce  witnesses that satisfy the specification non-vacuously, which can serve as  tests. Our approach avoids vacuities by requiring that certain faulty SUTs violate the specification.

\paragraph{Testing with a model checker}
Model checkers can be utilized to compute tests from temporal specifications~\cite{FraserWA09}.  The method by Fraser and Ammann~\cite{FraserA08} ensures that properties are not vacuously satisfied and that faults propagate to observable property violations (using finite-trace semantics for LTL).  Tan et al.~\cite{TanSL04} also define and apply a coverage  metric based on vacuity for LTL.  Ammann et al.~\cite{AmmannDX01} create tests  from CTL~\cite{ClarkeE81} specifications using model mutations.  All these  methods assume that a deterministic system model is available in addition to the  specification.  Fraser and Wotawa~\cite{FraserW07} also consider  non-deterministic models, but issue inconclusive verdicts if the system  deviates from the behavior foreseen in the test case.  In contrast, we search  for test strategies that achieve their goal for \emph{every} realization of the  specification.   Boroday et al.~\cite{BorodayPG07} aim for a similar guarantee (calling it  \emph{strong test cases}) using a model checker, but do not consider adaptive  test cases, and use a finite state machine as a specification. 

\paragraph{Synthesis of test strategies}
Bounded synthesis~\cite{FinkbeinerS13} aims for finding a system implementation of minimal size in the number of states.  Symbolic procedures based on binary decision diagrams~\cite{Ehlers12} and satisfiability solving~\cite{KhalimovJB13} exist.  In our setting, we do not synthesize an implementation of the system, but an adaptive test strategy, i.e., a controller that mimics the system's environment to enforce a certain test goal.  In contrast to a complete implementation of the controller, we strive for finding a partial implementation that assigns values only to those signals that necessarily contribute to reach the test goal.  Other signals can be kept non-deterministic and either chosen during execution of the test strategy or randomized.  We use a post-processing procedure that eliminates assignments from the test strategy and invokes a modelchecker to verify that the test goal is still enforced.  This post-processing step is conceptually similar to procedures that aim for counterexample simplification~\cite{JinRS04} and don't care identification in test patterns~\cite{MiyaseK04}.  Jin et al.~\cite{JinRS04} separate a counterexample trace into forced segments that unavoidably progress towards the specification violation and free segments that, if avoided, may have prevented the specification violation.  Our post-processing step is similar, but instead of counterexamples, adaptive test strategies are post-processed.  Miyase and Kajihara~\cite{MiyaseK04} present an approach to identify don't cares in test patterns of combinational circuits.  In contrast to combinational circuits, we deal with reactive systems.  Instead of post-processing a complete test strategy, a partial test strategy can be directly synthesized by modifying a synthesis procedure to compute minimum satisfying assignments~\cite{DilligDMA12}.  Although feasible, modifying a synthesis procedure requires a lot of work.  Our post-processing procedure uses the synthesis procedure in a plug-and-play fashion and does not require manual changes in the synthesis procedure.

\section{Preliminaries and Notation}
\label{sec:prelim}

\paragraph{Traces}
We want to test reactive systems that have a finite set $\inp=\{i_1,\ldots,i_m\}$ of Boolean inputs and a finite set $\outp=\{o_1,\ldots,o_n\}$ of Boolean outputs.  The input alphabet is $\inalp=2^\inp$, the output alphabet is $\outalp=2^\outp$, and $\alp=2^{\inp  \cup \outp}$.  An infinite word $\trace$ over $\alp$ is an \emph{(execution) trace} and the set $\alp^\omega$ is the set of all infinite words over $\alp$.

\paragraph{Linear Temporal Logic}
We use \emph{Linear Temporal Logic~(LTL)}~\cite{Pnueli77} as a specification language for reactive systems.  The syntax is defined as follows: every input or output $p\in\inp\cup \outp$ is an LTL formula; and if $\varphi_1$ and $\varphi_2$ are LTL formulas, then so are $\neg \varphi_1$, $\varphi_1 \vee \varphi_2$, $\nextt \varphi_1$ and $\varphi_1 \until \varphi_2$.  We write $\trace \models \varphi$ to denote that a trace $\trace = \letter_0 \letter_1 \ldots \in \alp^\omega$ \emph{satisfies} LTL formula $\varphi$.  This is defined inductively as follows:
\begin{compactitem}
\item $\letter_0 \letter_1 \letter_2 \ldots \models p$ iff $p \in \letter_0$,
\item $\trace \models \neg \varphi$ iff $\trace \not\models \varphi$,
\item $\trace \models \varphi_1 \vee \varphi_2$ iff $\trace \models \varphi_1$ or $\trace \models \varphi_2$,
\item $\letter_0 \letter_1 \letter_2 \ldots \models \nextt \varphi$ iff $\letter_1 \letter_2 \ldots \models \varphi$, and
\item $\letter_0 \letter_1  \ldots \models \varphi_1 \until \varphi_2$ iff $\exists j \ge 0 \scope \letter_j \letter_{j+1} \ldots \models \varphi_2 \wedge \forall 0 \leq k < j \scope \letter_k \letter_{k+1} \ldots \models \varphi_1$.
\end{compactitem}
That is, $\nextt \varphi$ requires $\varphi$ to hold in the \emph{next} step, and $\varphi_1 \until \varphi_2$ means that $\varphi_1$ must hold \emph{until} $\varphi_2$ holds (and $\varphi_2$ must hold eventually).  We also use the usual abbreviations $\varphi_1 \wedge \varphi_2 = \neg (\neg \varphi_1 \vee \neg \varphi_2)$, $\varphi_1 \rightarrow \varphi_2 = \neg \varphi_1 \vee \varphi_2$, $\eventually \varphi = \true \until \varphi$ (meaning that $\varphi$ must hold \emph{eventually}), and $\always \varphi = \neg \eventually \neg \varphi$ ($\varphi$ must hold \emph{always}).  By $\varphi[x \leftarrow y]$ we denote the LTL formula $\varphi$ where all occurrences of~$x$ have been textually replaced by~$y$.

\paragraph{Mealy machines}
We use Mealy machines to model the reactive system under test.  A \emph{Mealy machine} is a tuple $\sys = (Q, q_0, \inalp, \outalp, \delta, \lambda)$, where $Q$ is a finite set of states, $q_0\in Q$ is the initial state, $\delta: Q \times \inalp \rightarrow Q$ is a total transition function, and $\lambda: Q \times \inalp \rightarrow \outalp$ is a total output function.  Given the input trace $\intrace = x_0 x_1 \ldots \in \inalp^\omega$, $\sys$ produces the output trace $\outtrace = \sys(\intrace) = \lambda(q_0, x_0) \lambda(q_1, x_1) \ldots \in \outalp^\omega$, where $q_{i+1} = \delta(q_i, x_i)$ for all $i \ge 0$. That is, in every time step $i$, the Mealy machine reads the input letter $x_i\in \inalp$, responds with an output letter $\lambda(q_i, x_i) \in \outalp$, and updates its state to $q_{i+1} = \delta(q_i, x_i)$.  A Mealy machine can directly model synchronous hardware designs, but also other systems with inputs and outputs evolving in discrete time steps.  We write $\mealyset(\inp,\outp)$ for the set of all Mealy machines with inputs $\inp$ and outputs $\outp$.

\paragraph{Moore machines}
We use Moore machines to describe test strategies.  A \emph{Moore machine} is a special Mealy machine with $\forall q\in Q\scope \forall x,x'\in \inalp \scope \lambda(q,x) = \lambda(q,x')$.  That is, $\lambda(q,x)$ is insensitive to $x$, i.e., becomes a function $\lambda: Q \rightarrow \outalp$. This means that the input $x_i$ at step $i$ can affect the next state $q_{i+1}$ and thus the next output $\lambda(q_{i+1})$ but not the current output $\lambda(q_i)$.  We write $\mooreset(\inp,\outp)$ for the set of all Moore machines with inputs $\inp$ and outputs $\outp$.

\paragraph{Composition}
Given Mealy machines $\sys_1 = (Q_1, q_{0,1}, 2^\inp, 2^{\outp_1}, \delta_1, \lambda_1) \in \mealyset(\inp,\outp_1)$ and $\sys_2 = (Q_2, q_{0,2}, 2^{\inp\cup\outp_1}, 2^{\outp_2}, \delta_2, \lambda_2) \in \mealyset(\inp\cup\outp_1, \outp_2)$, we write $\sys = \sys_1 \concat \sys_2$ for their sequential composition $\sys = (Q_1 \times Q_2, (q_{0,1}, q_{0,2}), 2^\inp, 2^{\outp_1 \cup \outp_2},$ $ \delta, \lambda)$, where $\sys \in \mealyset(\inp,\outp_1\cup \outp_2)$ with $\delta\bigl((q_1, q_2), x\bigr) = \bigl(\delta_1(q_1,x), \delta_2(q_2, x \cup \lambda_1(q_1,x))\bigr)$ and $\lambda\bigl((q_1, q_2), x\bigr) = \lambda_1(q_1,x) \cup \lambda_2\bigl(q_2,x \cup \lambda_1(q_1,x)\bigr)$.  Note that $x \in 2^\inp$.

\paragraph{Systems and test strategies}
A \emph{reactive system} $\sys$ is a Mealy machine.  An \emph{(adaptive) test strategy} is a Moore machine $\strat = (T, t_0, \outalp, \inalp, \Delta, \Lambda)$ with input and output alphabet swapped.  That is, $\strat$ produces values for input signals and reacts to values of output signals.  A test strategy $\strat$ can be \emph{run} on a system $\sys$ as follows.  In every time step $i$ (starting with $i=0$), $\strat$ first computes the next input $x_i=\Lambda(t_i)$.  Then, the system computes the output $y_i = \lambda(q_i, x_i)$.  Finally, both machines compute their next state $t_{i+1} = \Delta(t_i, y_i)$ and $q_{i+1} = \delta(q_i, x_i)$.  We write $\trace(\strat,\sys) = (x_0 \cup y_0) (x_1 \cup y_1) \ldots \in \alp^\omega$ for the resulting execution trace.  If $\strat = (T, t_0, 2^{\outp'}, \inalp, \Delta, \Lambda) \in \mooreset(\outp', \inp)$ can observe only a subset $\outp'\subseteq \outp$ of the outputs, we define $\trace(\strat,\sys)$ with $t_{i+1} = \Delta(t_i, y_i \cap \outp')$. A \emph{test suite} is a set $\suite \subseteq \mooreset(\outp,\inp)$ of adaptive test strategies.

\paragraph{Realizability}
A Mealy machine $\sys \in \mealyset(\inp,\outp)$ \emph{realizes} an LTL formula $\varphi$, written $\sys \real \varphi$, if $\forall \moore \in \mooreset(\outp, \inp) \scope \trace(\moore,\sys) \models \varphi$.  An LTL formula $\varphi$ is \emph{Mealy-realizable} if there exists a Mealy machine that realizes it.  A Moore machine $\moore \in \mooreset(\inp, \outp)$ realizes $\varphi$, written $\moore\real\varphi$, if $\forall \sys \in \mealyset(\outp, \inp) \scope \trace(\moore,\sys) \models \varphi$.  A \emph{model checking procedure} checks if a given Mealy (Moore) machine $\sys$ ($\moore$) realizes an LTL specification $\varphi$ and returns $\true$ iff $\sys \real \varphi$ ($\moore \real \varphi$) holds.  We denote the call of a model checking procedure by $\modelcheck\bigl(\sys,\varphi\bigr)$ ($\modelcheck\bigl(\moore,\varphi\bigr)$).

\paragraph{Reactive synthesis}
We use reactive synthesis to compute test strategies.  A \emph{reactive (Moore, LTL) synthesis procedure} takes as input a set $\inp$ of Boolean inputs, a set $\outp$ of Boolean outputs, and an LTL specification $\varphi$ over these signals.  It produces a Moore machine $\moore \in \mooreset(\inp, \outp)$ that realizes $\varphi$, or the message \unreal if no such Moore machine exists.  We denote this computation by $\moore = \synt(\inp, \outp, \varphi)$.  A \emph{synthesis procedure with partial information} is defined similarly, but takes a subset $\inp' \subseteq \inp$ of the inputs as an additional argument.  As output, the synthesis procedure produces a Moore machine $\moore' = \syntp(\inp, \outp, \varphi, \inp')$ with $\moore' \in \mooreset(\inp', \outp)$ that realizes $\varphi$ while only observing the inputs $\inp'$, or the message 
\unreal if no such Moore machine exists.  We assume that both synthesis procedure, $\synt$ and $\syntp$, can be called \emph{incrementally} with an additional parameter $\constr$, where $\constr$ denotes a set of Moore machines.  The incremental synthesis procedures $\moore = \synt(\inp, \outp, \varphi, \constr)$ and $\moore' = \syntp(\inp, \outp, \varphi, \inp', \constr)$ compute Moore machines $\moore$ and $\moore^\prime$, respectively, as before but with the additional constraints that $\moore, \moore^\prime \not\in \constr$.

\paragraph{Fault versus failure}
A Mealy machine $\sys \in \mealyset(\inp,\outp)$ is \emph{faulty} with respect to LTL formula $\varphi$ (specification) iff $\sys \not\real \varphi$, i.e., $\exists \moore \in \mooreset(\outp, \inp) \scope \trace(\moore,\sys) \not\models \varphi$.  We call a trace $\trace(\moore,\sys)$ that uncovers a faulty behavior of $\sys$ a \emph{failure} and a deviation between $\sys$ and any correct realization $\sys^\prime$, i.e., $\sys^\prime \real \varphi$, a \emph{fault}.  For a fixed faulty $\sys$, there are multiple correct $\sys^\prime$ that realize $\varphi$ and thus a fault in $\sys$ can be characterized by multiple, different ways.  As a simplification, we assume that in practice every faulty $\sys$ is close to a correct $\sys^\prime$ and only deviates in a simple fault.  In the next section, we will show how this idea can be leveraged to determine test suites independent of the implementation and the concrete fault manifestation.

\section{Synthesis of Adaptive Test Strategies}
\label{sec:synt}

This section presents our approach for synthesizing adaptive test strategies for reactive systems specified in LTL.  First, we elaborate on the coverage objective we aim to achieve.  Then we present our strategy synthesis algorithm.  Finally, we discuss extensions and variants of the algorithm.

\subsection{Coverage Objective for Test Strategy Computation}
\label{sec:cov}

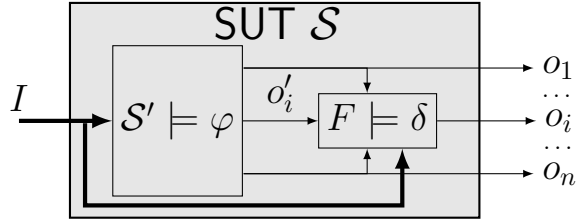
\begin{figure}[t]
  \centering
  \begin{tikzpicture}[>=latex,->]
  \node[fill=gray!20!white,draw,minimum height=2cm]   (sys)   {\Large $\sys^\prime \models \varphi$};
  \node[fill=gray!20!white,draw, right=1cm of sys]    (fault) {\Large $F \models \delta$};
  
  \draw[->]             (sys)   -- node [above] {\Large $o_i^\prime$} (fault);
  \draw[<-,ultra thick] (sys)   --++ (left:35pt) node [yshift=3pt] (i) {} --++ (left:25pt) node [above] {\Large $\inp$};

  \draw[ultra thick]    (i)     --++ (down:35pt) node {} --++ (right:120pt) node (A) {} -- (A |- fault.south);
  
  \draw[->]             (fault) --++ (right:60pt) node [right] {\Large $o_i$};

  \node (B) at ([xshift=-3pt,yshift= 20pt] sys.east) {};
  \node (C) at ([xshift=-3pt,yshift=-20pt] sys.east) {};

  \draw (B) --++ (right:113pt) node [right] (o1) {\Large $o_1$};
  \node [below=-.8pt of o1] {\dots};
  \draw (C) --++ (right:113pt) node [right] (on) {\Large $o_n$};
  \node [above=-1.5pt of on] {\dots};

  \draw (B) --++ (right:50pt) node (D) {} -- (D |- fault.north);
  \draw (C) --++ (right:50pt) node (E) {} -- (E |- fault.south);
  
  \begin{scope}[on background layer]
  \path let \p1=(sys.south), \p2=(fault.east) in node (helper) at (\x2,\y1) {};
  \draw [thick,fill=gray!20!white] ([xshift=-16pt,yshift=16pt] sys.north west) rectangle ([xshift=16pt,yshift=-8pt] helper);
  \end{scope}
  
  \node at ([yshift=37pt] $(sys)!.5!(fault)$) {\LARGE {\sf SUT} $\sys$};
\end{tikzpicture}
  \caption{Coverage goal illustration for fault.}
  \label{fig:tg}
\end{figure}

Many coverage metrics~\cite{Mathur08} exist to assess the quality of a test suite.  Since the goal in testing is to detect bugs, we follow a fault-centered approach: a test suite has high quality if it reveals certain kinds of faults in a system. As illustrated in \reffig{tg}, we assume that our SUT is ``almost correct'', i.e., it is composed of a correct implementation $\sys'$ of the specification $\varphi$, but with a fault $F$ that affects one of the outputs.  In order to make our approach flexible, we allow the user to define the considered faults as an LTL formula $\faultform$.  Through $\faultform$, the user can define both permanent and transient faults of various types.  For instance, $\faultform = \eventually(o_i \leftrightarrow \neg o_i')$ describes a bit-flip that occurs at least once, $\always\eventually \neg o_i$ models a stuck-at-0 fault that occurs infinitely often, and $\always(\nextt(o_i) \leftrightarrow o_i')$ models a permanent shift by one time step.  We strive for a test suite that reveals \emph{every} fault that satisfies $\faultform$ for \emph{every} realization of $\varphi$.  This renders the test suite independent of the implementation and the concrete fault manifestation.  The following definition formalizes this intuition into a coverage objective.

\begin{definition}\label{def:complete}
A test suite $\suite \subseteq \mooreset(\outp,\inp)$ for a system with inputs $\inp$, outputs $\outp$, and specification $\varphi$ is \emph{universally complete}\footnote{The word ``complete'' indicates that every considered fault is revealed at every output. The word ``universal'' indicates that this is achieved for every (otherwise correct) system.} with respect to a given fault model $\faultform$ iff
\begin{multline}
  \forall o_i\in \outp \scope \forall \sys' \in \mealyset(\inp,\outp \cup \{o_i'\} \setminus \{o_i\}) \scope \\
  \forall F \in \mealyset(\inp \cup \outp \cup \{o_i'\} \setminus \{o_i\}, \{o_i\}) \scope \exists \strat \in \suite \scope\\
  \Bigl(\bigl(\sys' \real \varphi[o_i\leftarrow o_i'] \wedge F \real \faultform\bigr)\rightarrow \bigl(\trace(\strat, \sys' \concat F) \not \models \varphi\bigr)\Bigr).
  \label{eq:complete_eq}
\end{multline}
\end{definition}

That is, for every output $o_i$, system $\sys'\real \varphi[o_i\leftarrow o_i']$, and fault $F \real \faultform$, $\suite$ must contain a test strategy $\strat$ that reveals the fault by causing a specification violation (\reffig{tg}).  Note that the test strategies $\strat \in \suite \subseteq \mooreset(\outp,\inp)$ cannot observe the signal $o_i'$.  The reason is that this signal $o_i'$ does not exist in the real system implementation(s) on which we run our tests --- it was only introduced to define our coverage objective.

There can be an unbounded number of system realizations $\sys' \real \varphi[o_i\leftarrow o_i']$ and faults $F \real \faultform$.  Computing a separate test strategy for each combination is thus not a viable option.  We rather strive for computing only one test strategy per output variable.
\begin{theorem}\label{th:exists}
A universally complete test suite $\suite \subseteq \mooreset(\outp,\inp)$ with respect to fault model $\faultform$ exists for a system with inputs $\inp$, outputs $\outp$, and specification $\varphi$ if
\begin{multline}
  \forall o_i\in \outp \scope \exists \strat \in \mooreset(\outp,\inp) \scope \forall \sys \in \mealyset(\inp,\outp \cup \{o_i'\}) \scope \\
  \trace(\strat, \sys) \models \bigl((\varphi[o_i\leftarrow o_i'] \wedge \faultform) \rightarrow \neg \varphi \bigr).
  \label{eq:theo_eq}
\end{multline}
\end{theorem}
\begin{proof}
\refeq{theo_eq} implies
\begin{multline}
  \forall o_i\in \outp \scope \forall \sys \in \mealyset(\inp,\outp \cup \{o_i'\}) \scope \exists \strat \in \mooreset(\outp,\inp) \scope \\
  \bigl(\sys \real \varphi[o_i\leftarrow o_i'] \land \faultform \bigr) \rightarrow \bigl(\trace(\strat, \sys) \not \models \varphi\bigr)
  \label{eq:rewi1}
\end{multline}
because (a) going from $\exists \strat \forall \sys$ to $\forall \sys \exists \strat$ can only make the formula weaker, and (b) $\sys \real \varphi[o_i\leftarrow o_i'] \land \faultform$ implies $\trace(\strat, \sys) \models \varphi[o_i\leftarrow o_i'] \land \faultform$ for all $\strat$, which can only make the left side of the implication stronger.  In turn, \refeq{rewi1} is equivalent to
\begin{multline}
  \forall o_i\in \outp \scope \forall \sys' \in \mealyset(\inp,\outp \cup \{o_i'\} \setminus \{o_i\}) \scope \\
  \forall F \in \mealyset(\inp \cup \outp \cup \{o_i'\} \setminus \{o_i\}, \{o_i\}) \scope \exists \strat \in \mooreset(\outp,\inp) \scope \\
  \bigl(\sys' \real \varphi[o_i\leftarrow o_i'] \land F \real \faultform\bigr) \rightarrow \bigl(\trace(\strat, \sys' \concat F) \not \models \varphi\bigr)
  \label{eq:rewi2}
\end{multline}
because for a given $\sys' \real \varphi[o_i\leftarrow o_i']$ and $F \real \faultform$ from \refeq{rewi2} we can define an equivalent system $\sys=(\sys' \concat F) \in \mealyset(\inp,\outp \cup \{o_i'\})$ for \refeq{rewi1} such that $\sys \real \varphi[o_i\leftarrow o_i'] \land \faultform$ is satisfied.  Also, for a given $\sys \real \varphi[o_i\leftarrow o_i'] \land \faultform$ from \refeq{rewi1} we can define a corresponding $\sys' \real \varphi[o_i\leftarrow o_i']$ and $F \real \faultform$ by stripping off different outputs.
\end{proof}

Theorem~\ref{th:exists} states that \refeq{theo_eq} is a sufficient condition for a universally complete test suite to exist.  If it were also a necessary condition, then computing one test strategy per output signal would be enough.  Unfortunately, this is not the case in general.

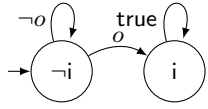
\begin{figure}[t]
  \centering
  \begin{tikzpicture}[>=latex,->,auto,initial text={},initial distance=3mm]
\node[state,initial,inner sep=0]  at  (0,0)       (S0) {$\neg \mathsf{i}$};
\node[state,inner sep=0]          at  (1.5,0)     (S1) {$\mathsf{i}$};

\path
(S0) edge [loop above]  node[xshift=-4mm,yshift=-4mm] {$\neg o$} (S0)
(S1) edge [loop above]  node[xshift=-5mm,yshift=-4mm] {$\true$} (S1)
(S0) edge [bend left]   node[xshift=0mm,yshift=-0.5mm] {$o$} (S1)
;
\end{tikzpicture}
  \caption{Test strategy $\strat_5$.}
  \label{fig:strategy incomplete}
\end{figure}

\begin{example}\label{ex:incompl}
Consider a system with input $\inp=\{i\}$, output $\outp=\{o\}$, and specification $\varphi = \bigl( \always(i \rightarrow \always i) \wedge \eventually i \bigr) \rightarrow \bigl( \always(o \rightarrow \always o) \wedge \eventually o \wedge \always(i \vee \neg o) \bigr)$.  The left side of the implication assumes that the input $i$ is set to $\true$ at some point, after which $i$ remains $\true$.  The right side requires the same for the output $o$.  In addition, $o$ must not be raised while $i$ is still $\false$.  This specification is realizable (e.g., by always setting $o=i$).  The test suite $\suite = \{\strat_5\}$ with $\strat_5$ shown in \reffig{strategy incomplete} is universally complete with respect to fault model $\faultform = \eventually(o \leftrightarrow \neg o')$, which requires the output to flip at least once: as long as $i$ is $\false$, any correct system implementation $\sys'\in \mealyset(\{i\},\{o'\}) \real \varphi[o_i\leftarrow o_i']$ must keep the output $o'=\false$.  Eventually, $F\real \faultform$ must flip the output $o$ to $\true$.  When this happens, $i$ is set to $\true$ by $\strat_5$ so that the resulting trace $\trace(\strat, \sys' \concat F)$ violates $\varphi$.  Still, \refeq{theo_eq} is $\false$\footnote{This is (at least partially) confirmed by our test strategy synthesis tool: it reports that no test strategy with less than $12$ states can satisfy \refeq{theo_eq}.}.  Strategy $\strat_5$ does not satisfy \refeq{theo_eq} because for the system $\sys \in \mealyset(\{i\},\{o,o'\})$ that sets $o'=\true$ and $o=\false$ in all time steps, we have $\trace(\strat_5, \sys) \models \bigl(\varphi[o_i\leftarrow o_i'] \wedge \faultform \wedge \varphi \bigr)$.  The reason is that $i$ stays $\false$, so $\varphi[o_i\leftarrow o_i']$ and $\varphi$ are vacuously satisfied by $\trace(\strat_5, \sys)$.  The formula $\faultform$ is satisfied because $o \leftrightarrow \neg o'$ holds in all time steps. Thus, $\sys$ is a counterexample to $\strat_5$ satisfying \refeq{theo_eq}.  Similar counterstrategies exist for all other test strategies.
\end{example}

The fact that \refeq{theo_eq} is not a necessary condition for a universally complete test suite to exist is somewhat surprising, especially in the light of the following two lemmas.  Based on these lemmas, the subsequent propositions will show that \refeq{theo_eq} is both sufficient and necessary (i.e., one test per output is enough) for many interesting cases.

\begin{lemma}\label{th:det}
For every LTL specification $\psi$ over some inputs $\inp$ and outputs $\outp$, we have that $\exists \strat \in  \mooreset(\outp, \inp) \scope \forall \sys \in \mealyset(\inp, \outp) \scope \trace(\strat, \sys) \models \psi$ holds if and only if $\forall \sys \in \mealyset(\inp, \outp) \scope \exists \strat \in  \mooreset(\outp, \inp) \scope \trace(\strat, \sys) \models \psi$ holds.
\end{lemma}

\begin{proof}
Synthesis from LTL specifications under complete information is (finite memory) determined~\cite{Martin75}, which means that either $\exists \strat \in  \mooreset(\outp, \inp) \scope \forall \sys \in \mealyset(\inp, \outp) \scope \trace(\strat, \sys) \models \psi$ or $\exists \sys \in \mealyset(\inp, \outp) \scope \forall \strat \in \mooreset(\outp, \inp) \scope \trace(\strat, \sys) \models \neg  \psi$ holds, but not both.  Less formal we can say that either there exists a test strategy $\strat$ that satisfies $\psi$ for all systems $\sys$, or there exists a system $\sys$ that can violate $\psi$ for all test strategies $\strat$.  From that, it follows that
\begin{eqnarray*}
&              & \exists \strat \in  \mooreset(\outp, \inp) \scope \forall \sys \in \mealyset(\inp, \outp) \scope \trace(\strat, \sys) \models \psi \\
& \text{ iff } & \neg  \exists \sys \in \mealyset(\inp, \outp) \scope \\
&              & \forall \strat \in  \mooreset(\outp, \inp) \scope \trace(\strat, \sys) \models \neg \psi \\
& \text{ iff } &
\forall \sys \in \mealyset(\inp, \outp) \scope \exists \strat \in  \mooreset(\outp, \inp) \scope \trace(\strat, \sys) \models \psi.
\end{eqnarray*}
\end{proof}

\begin{lemma}\label{th:impl}
For all LTL specifications $A,G$ over inputs $\inp$ and outputs $\outp$, we have that
\begin{eqnarray}
\begin{split}
&& \forall \sys \in \mealyset(\inp, \outp) \scope \exists \strat \in  \mooreset(\outp, \inp) \scope \\
&& (\sys \real A) \rightarrow \bigl(\trace(\strat, \sys) \models G\bigr)
\end{split}
\label{eq:impl1}
\\
\begin{split}
&\text{ iff } & \forall \sys \in \mealyset(\inp, \outp) \scope \exists \strat \in  \mooreset(\outp, \inp) \scope \\
&             & \trace(\strat, \sys) \models (A \rightarrow G).
\end{split}
\label{eq:impl2}
\end{eqnarray}
\end{lemma}

\begin{proof}
Direction $\Rightarrow$: We show that \refeq{impl2} being $\false$ contradicts with \refeq{impl1} being $\true$.
\begin{eqnarray*}
&              & \neg \forall \sys \in \mealyset(\inp, \outp) \scope \exists \strat \in  \mooreset(\outp, \inp) \scope \\
&              & \trace(\strat, \sys) \models (A \rightarrow G) \\
& \text{ iff } & \exists \sys \in \mealyset(\inp, \outp) \scope \forall \strat \in  \mooreset(\outp, \inp) \scope \\
&              & \trace(\strat, \sys) \models (A \land \neg  G) \\
& \text{ iff } & \exists \sys \in \mealyset(\inp, \outp) \scope \sys \real (A \land \neg  G), \text{ which implies } \\
&              & \exists \sys \in \mealyset(\inp, \outp) \scope \forall \strat \in  \mooreset(\outp, \inp) \scope \\
&              & (\sys \real A) \land \bigl(\trace(\strat, \sys) \models \neg  G\bigr).
\end{eqnarray*}
Direction $\Leftarrow$: Using the LTL semantics, we can rewrite $\trace(\strat, \sys) \models (A \rightarrow G)$ in \refeq{impl2} as $\bigl(\trace(\strat, \sys) \models A\bigr) \rightarrow \bigl(\trace(\strat, \sys) \models G\bigr)$.  Since $\sys \real A$ implies $\trace(\strat', \sys) \models A$ for every $\strat' \in \mooreset(\inp, \outp)$, the assumption in \refeq{impl1} is not weaker, so \refeq{impl1} is not stronger.
\end{proof}

These two lemmas state that quantifiers can be swapped and that assuming $\trace(\strat, \sys) \models A$ is equivalent to assuming $(\sys \real A)$ for the case where $\strat$ has full information about the outputs of $\sys$.  Yet, in our setting, test strategies $\strat \in \mooreset(\outp,\inp)$ have incomplete information about the system $\sys \in \mealyset(\inp,\outp \cup \{o_i'\})$ because they cannot observe $o_i'$.  Still, $\strat$ must enforce $(\varphi[o_i\leftarrow o_i'] \wedge \faultform) \rightarrow \neg \varphi,$ which refers to this hidden signal.  Thus,~\reflem{det} and~\ref{th:impl} cannot be applied to \refeq{theo_eq} in general.  However, in cases where there is (effectively) no hidden information, the lemmas can be used to prove that \refeq{theo_eq} is both a necessary and a sufficient condition for a universally complete test suite to exist.  The following propositions show that this holds for many cases of practical interest.

The intuitive reason is that $\varphi[o_i\leftarrow o_i']$ can be rewritten to $\varphi[o_i\leftarrow \psi]$ in \refeq{theo_eq}, which eliminates the hidden signal such that~\reflem{det} and~\ref{th:impl} can be applied.

\begin{proposition}\label{th:worksfixed}
Given a fault model of the form $\faultform = \always(o_i' \leftrightarrow \psi)$, where $\psi$ is an LTL formula over $\inp$ and $\outp$, a universally complete test suite $\suite \subseteq \mooreset(\outp,\inp)$ with respect to $\faultform,\inp,\outp$, and $\varphi$ exists if and only if \refeq{theo_eq} holds.
\end{proposition}

\begin{proof}
$\varphi[o_i\leftarrow o_i'] \land \always(o_i' \leftrightarrow \psi)$ is equivalent to $\varphi[o_i\leftarrow \psi] \land \always(o_i' \leftrightarrow \psi)$.  Thus, \refeq{theo_eq} becomes
\begin{multline*}
  \forall o_i\in \outp \scope \exists \strat \in \mooreset(\outp,\inp) \scope \forall \sys \in \mealyset(\inp,\outp \cup \{o_i'\}) \scope \\
  \trace(\strat, \sys) \models \bigl((\varphi[o_i\leftarrow \psi] \land \always(o_i' \leftrightarrow \psi)) \rightarrow \neg  \varphi \bigr),
\end{multline*}
which is equivalent to
\begin{multline*}
  \forall o_i\in \outp \scope \exists \strat \in \mooreset(\outp,\inp) \scope \forall \sys \in \mealyset(\inp,\outp) \scope \\
  \trace(\strat, \sys) \models \bigl(\varphi[o_i\leftarrow \psi] \rightarrow \neg  \varphi \bigr)
 \end{multline*}
Because of the $\always$ operator, a unique value for $o_i'$ exist in all time steps and thus, $o_i'$ is just an abbreviation for $\psi$.  Whether this abbreviation $o_i'$ is available as output of $\sys$ or not is irrelevant, because $\strat$ cannot observe $o_i'$ anyway.  Since $o_i'$ no longer occurs, \reflem{det} and \reflem{impl} can be applied to prove equivalence between \refeq{theo_eq} and
\begin{multline*}
  \forall o_i\in \outp \scope \forall \sys \in \mealyset(\inp,\outp) \scope \exists \strat \in \mooreset(\outp,\inp) \scope \\
  (\sys \real \varphi[o_i\leftarrow \psi]) \rightarrow \trace(\strat, \sys) \not\models \varphi.
\end{multline*}
As $\strat$ cannot observe $o_i'$, it is irrelevant whether the truth value of $\psi$ is available as additional output $o_i'$ of $\sys$ or not.  Hence, the above formula is equivalent to
\begin{multline*}
  \forall o_i\in \outp \scope \forall \sys \in \mealyset(\inp,\outp\cup\{o_i'\}) \scope \exists \strat \in \mooreset(\outp,\inp) \scope \\
  (\sys \real (\varphi[o_i\leftarrow \psi] \land \always(o_i' \leftrightarrow \psi)) \rightarrow \trace(\strat, \sys) \not\models \varphi
\end{multline*}
and
\begin{multline*}
  \forall o_i\in \outp \scope \forall \sys \in \mealyset(\inp,\outp\cup\{o_i'\}) \scope \exists \strat \in \mooreset(\outp,\inp) \scope \\
  (\sys \real (\varphi[o_i\leftarrow o_i'] \land \delta) \rightarrow \trace(\strat, \sys) \not\models \varphi,
\end{multline*}
i.e., to \refeq{rewi1}.  The remaining steps can be taken from the proof of Theorem~\ref{th:exists}.
\end{proof}

Proposition~\ref{th:worksfixed} entails that computing one test strategy per output $o_i\in\outp$ is enough for fault models such as permanent bit flips (defined by $\faultform = \always(o_i' \leftrightarrow \neg o_i)$).

\begin{proposition}\label{th:worksstuck}
If the fault model $\faultform$ does not reference $o_i'$, a universally complete test suite $\suite \subseteq \mooreset(\outp,\inp)$ with respect to $\faultform,\inp,\outp$, and $\varphi$ exists iff \refeq{theo_eq} holds.
\end{proposition}

\begin{proof}
We show that \refeq{theo_eq} holds if and only if \refeq{rewi1} holds.  The remaining steps have already been proven for Theorem~\ref{th:exists}.

\begin{lemma}\label{th:aux1}
\refeq{theo_eq} holds if and only if
\begin{equation}
\begin{split}
\forall o_i\in \outp \scope \exists \strat \in \mooreset(\outp,\inp) \scope \forall \sys \in \mealyset(\inp,\outp) \scope \\
\trace(\strat, \sys) \models (\faultform \rightarrow \neg  \varphi).
\label{eq:sap0}
\end{split}
\end{equation}
\end{lemma}

\begin{proof}
Direction $\Leftarrow$ is obvious because \refeq{theo_eq} contains stronger assumptions (and $\forall \sys \in \mealyset(\inp,\outp)$ can be changed to $\forall \sys \in \mealyset(\inp,\outp\cup\{o_i'\})$ in \refeq{sap0} because $\faultform \rightarrow \neg \varphi$ does not contain $o_i'$).

\noindent Direction $\Rightarrow$: We show that \refeq{sap0} being $\false$ contradicts with \refeq{theo_eq} being $\true$.
\begin{eqnarray}
&&
\neg \forall o_i\in \outp \scope \exists \strat \in \mooreset(\outp,\inp) \scope \notag \\ &&
\forall \sys \in \mealyset(\inp,\outp) \scope \trace(\strat, \sys) \models (\faultform \rightarrow \neg  \varphi)
\label{eq:sap01}
\allowdisplaybreaks \\
&\text{iff }&
\exists o_i\in \outp \scope \forall \strat \in \mooreset(\outp,\inp) \scope \notag \\ &&
\exists \sys \in \mealyset(\inp,\outp) \scope \trace(\strat, \sys) \models (\faultform \land \varphi)
\label{eq:sap02}
\allowdisplaybreaks \\
&\text{iff }&
\exists o_i\in \outp \scope \exists \sys \in \mealyset(\inp,\outp) \scope \notag \\ &&
\forall \strat \in \mooreset(\outp,\inp) \scope \trace(\strat, \sys) \models (\faultform \land \varphi)
\label{eq:sap03}
\allowdisplaybreaks \\
&\text{iff }&
\exists o_i\in \outp \scope \exists \sys \in \mealyset(\inp,\outp) \scope \sys \real (\faultform \land \varphi)
\label{eq:sap04}
\allowdisplaybreaks \\
&\text{iff }&
\exists o_i\in \outp \scope \exists \sys' \in \mealyset(\inp,\outp\cup\{o_i'\}) \scope \notag \\ &&
\sys' \real (\varphi[o_i\leftarrow o_i'] \land \faultform \land \varphi),
\label{eq:sap05}
\allowdisplaybreaks \\
&\text{iff }&
\exists o_i\in \outp \scope \exists \sys' \in \mealyset(\inp,\outp\cup\{o_i'\}) \scope \notag \\ &&
\forall \strat \in \mooreset(\outp\cup\{o_i'\},\inp) \scope \notag\\&&
\trace(\strat, \sys) \models (\varphi[o_i\leftarrow o_i'] \land \faultform \land \varphi),
\label{eq:sap06}
\allowdisplaybreaks \\
&\text{iff }&
\exists o_i\in \outp \scope \forall \strat \in \mooreset(\outp\cup\{o_i'\},\inp) \scope \notag \\ &&
\exists \sys' \in \mealyset(\inp,\outp\cup\{o_i'\}) \scope \notag\\&&
\trace(\strat, \sys) \models (\varphi[o_i\leftarrow o_i'] \land \faultform \land \varphi),
\label{eq:sap07}
\allowdisplaybreaks \\
&\Rightarrow&
\exists o_i\in \outp \scope \forall \strat \in \mooreset(\outp,\inp) \scope \notag \\ &&
\exists \sys' \in \mealyset(\inp,\outp\cup\{o_i'\}) \scope \notag\\&&
\trace(\strat, \sys) \models(\varphi[o_i\leftarrow o_i'] \land \faultform \land \varphi),
\label{eq:sap08}
\end{eqnarray}
which contradicts \refeq{theo_eq}.  (\ref{eq:sap02})$\Leftrightarrow$(\ref{eq:sap03}) holds because of \reflem{det} and (\refeq{sap04})$\Leftrightarrow$(\refeq{sap05}) holds because $\faultform \land \varphi$ does not contain $o_i'$, so $\sys'$ can be $\sys$ with $o_i' \leftrightarrow o_i$.  (\refeq{sap06})$\Leftrightarrow$(\refeq{sap07}) holds because of \reflem{det}.  Finally, (\refeq{sap07}) implies (\refeq{sap08}) because $\strat$ has less information in (\refeq{sap08}).\\
\end{proof}

\begin{lemma}\label{th:aux2}
\refeq{sap0} holds if and only if \refeq{rewi1} holds.
\end{lemma}

\begin{proof}
Direction $\Rightarrow$: is obvious because \refeq{sap0} is equivalent to \refeq{theo_eq} (\reflem{aux1}) and \refeq{theo_eq} implies \refeq{rewi1} (see proof for Theorem~\ref{th:exists}). \\
Direction $\Leftarrow$: we show that \refeq{sap0} being $\false$ contradicts \refeq{rewi1} being $\true$.  \refeq{sap0} being $\false$ implies \refeq{sap05} (see above).  As $\sys' \real (\varphi[o_i\leftarrow o_i'] \land \faultform \land \varphi)$ implies $(\sys' \real \varphi[o_i\leftarrow o_i'] \land \faultform) \land \bigl(\trace(\strat, \sys)\models \varphi \bigl)$ for all $\strat \in \mooreset(\outp\cup\{o_i'\},\inp)$ and thus also for all $\strat \in \mooreset(\outp,\inp)$, \refeq{rewi1} cannot hold.
\end{proof}
\end{proof}

Thus, the assumption $\sys' \real \varphi[o_i\leftarrow o_i']$ can be dropped from \refeq{complete_eq} if the fault model does not reference $o_i'$.  Correspondingly, $\trace(\strat, \sys) \models \bigl((\varphi[o_i\leftarrow o_i'] \wedge \faultform) \rightarrow \neg \varphi \bigr)$ simplifies to $\trace(\strat, \sys) \models (\faultform \rightarrow \neg \varphi)$ in \refeq{theo_eq}.  Since $o_i'$ is now gone, \reflem{det} and~\ref{th:impl} apply. In general, the assumption $\sys' \real \varphi[o_i\leftarrow o_i']$ is needed to prevent a faulty system $\sys'\not\real\varphi[o_i\leftarrow o_i']$ from compensating the fault $F\real\faultform$ such that $\sys'\concat F \real \varphi$.  E.g., for $\inp=\emptyset$, $\outp=\{o\}$, $\varphi=\always o$ and $\faultform = \always(o \leftrightarrow \neg o')$, \refeq{complete_eq} would be $\false$ without $\sys'\real\varphi[o_i\leftarrow o_i']$ because there exists an $\sys'$ that always sets $o'=\false$, in which case $\sys' \concat F$ has $o$ correctly set to $\true$.  However, if $\faultform$ does not reference $o'$, such a fault compensation is not possible.

Proposition~\ref{th:worksstuck} applies to permanent or transient stuck-at-0 or stuck-at-1 faults (e.g., $\faultform=\eventually \neg o_i$ or $\faultform=\always\eventually o_i$), but also to faults where $o_i$ keeps its previous value (e.g., $\faultform=\eventually(o_i\leftrightarrow \nextt(o_i)$) or takes the value of a different input or output (e.g., $\faultform=\always\eventually(o_i \leftarrow i_3)$).  Together with Proposition~\ref{th:worksfixed}, it shows that computing one test strategy per output is enough for many interesting fault models.  Finally, even if neither Proposition~\ref{th:worksfixed} nor Proposition~\ref{th:worksstuck} applies, computing one test strategy per output may still suffice for the concrete $\varphi$ and $\faultform$ at hand.  In the next section, we thus rely on \refeq{theo_eq} to compute one test strategy per output in order to obtain universally complete test suites.

\subsection{Test Strategy Computation}
\label{sec:comp}

\paragraph{Basic idea}
Our test case generation approach builds upon Theorem~\ref{th:exists}: for every output $o_i\in\outp$, we want to find a test strategy $\strat_i \in \mooreset(\outp,\inp)$ such that $\forall \sys \in \mealyset(\inp,\outp \cup \{o_i'\}) \scope \trace(\strat_i, \sys) \models \bigl((\varphi[o_i\leftarrow o_i'] \wedge \faultform) \rightarrow \neg \varphi \bigr)$ holds.  Recall from Section~\ref{sec:prelim} that a synthesis procedure $\moore = \syntp(\inp, \outp, \psi, \inp',\constr)$ with partial information computes a Moore machine $\moore\in \mooreset(\inp', \outp) \setminus \constr$ with $\inp'\subseteq \inp$ such that a certain LTL objective $\psi$ is enforced in all environments, i.e., $\forall \sys \in \mealyset(\outp, \inp) \scope \trace(\moore,\sys) \models \psi$.
If no such $\moore$ exists, $\syntp$ returns \unreal.  Also recall that a test strategy is a Moore machine with input and output signals swapped.  We can thus call $\strat_i := \synt_p\bigl(\outp \cup \{o_i'\}, \inp, (\varphi[o_i\leftarrow o_i'] \wedge \faultform) \rightarrow \neg \varphi, \outp,\constr \bigr)$ for every output $o_i \in \outp$ in order to obtain a universally complete test suite with respect to fault model $\faultform$ for a system with inputs $\inp$, outputs $\outp$, and specification $\varphi$.  If $\synt_p$ succeeds (does not return \unreal) for all $o_i \in \outp$, the resulting test suite $\suite = \{\strat_i \mid o_i\in \outp\}$ is guaranteed to be universally complete.  However, since Theorem~\ref{th:exists} only gives a sufficient but not a necessary condition, this procedure may fail to find a universally complete test suite, even if one exists, in general.  In cases where Proposition~\ref{th:worksfixed} or Proposition~\ref{th:worksstuck} applies, it is both sound and complete, though.

\begin{figure}[t]
  \centering
  \includegraphics[width=1.0\textwidth]{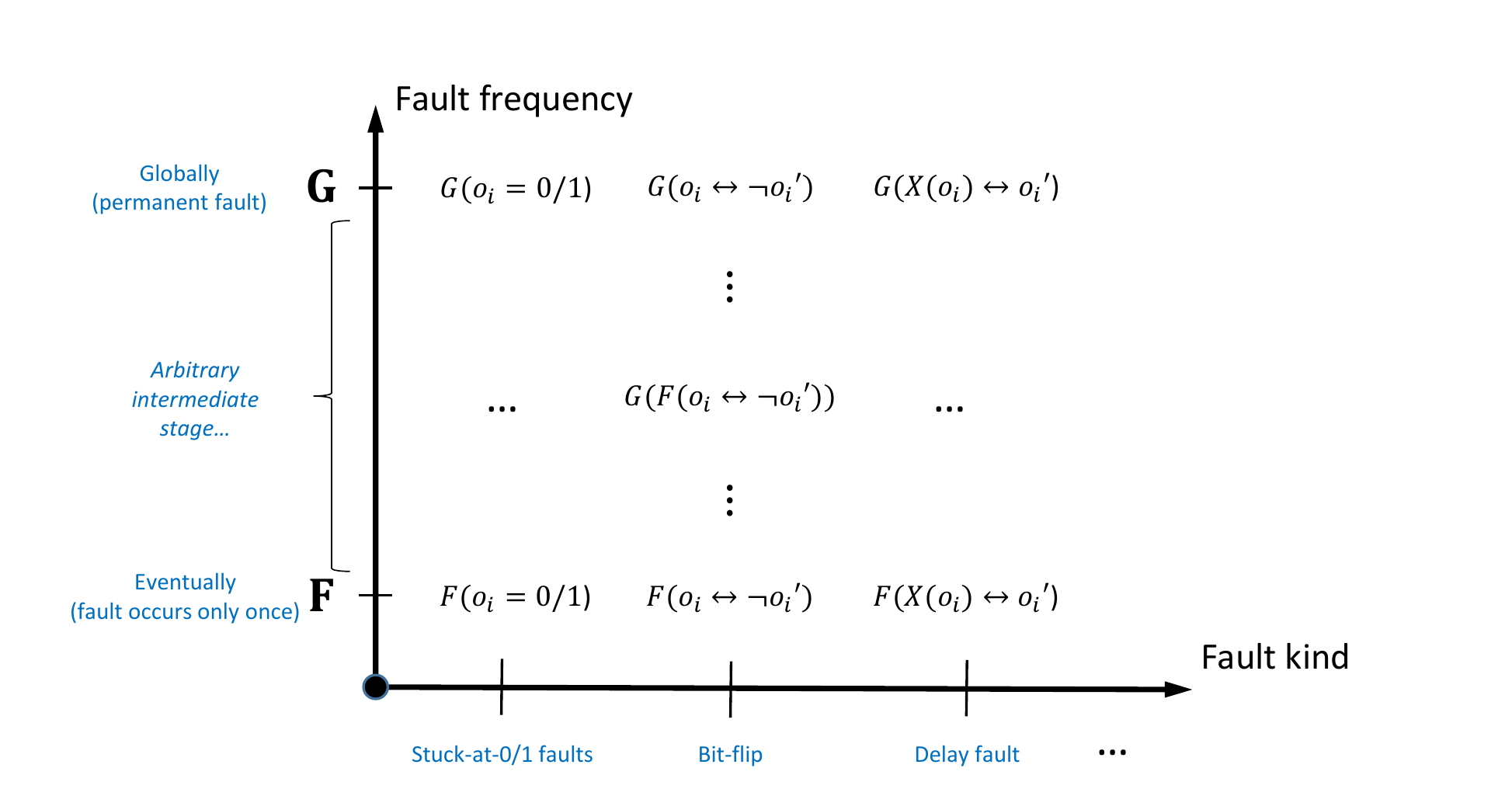}
  \caption{Relationship between fault kind and fault frequency.}
  \label{fig:faultmodelfreqdia}
\end{figure}

\paragraph{Fault models}
In order to simplify the user input, we split the fault model $\faultform$ in our coverage objective from Definition~\ref{def:complete} into two parts: the fault kind $\faultkind$ and the fault frequency $\faultfreq$ (\reffig{faultmodelfreqdia} illustrates the relationship).  The fault kind $\faultkind$ is an LTL formula that is given by the user and defines \emph{which} faults we consider.  For instance, $\faultkind=\neg o_i$ describes a stuck-at-0 fault, $\faultkind=o_i \leftrightarrow \neg o_i'$ defines a bit-flip, and $\faultkind=o_i' \leftrightarrow \nextt(o_i)$ describes a delay by one time step.  The fault frequency $\faultfreq$ describes \emph{how often} a fault of the specified kind occurs, and is chosen by our algorithm, unless it is specified by the user.  We distinguish $4$ fault frequencies, which we describe using temporal LTL operators.
\begin{compactitem}
\item Fault frequency $\always$ means that the fault is permanent.
\item Frequency $\eventually \! \always$ means that the fault occurs from some time step $i$ on permanently.  Yet, we do not make any assumptions about the precise value of $i$.
\item Frequency $\always \! \eventually$ states that the fault strikes infinitely often, but not when exactly.
\item Frequency $\eventually$ means that the fault occurs at least once.
\end{compactitem}
The fault model $\faultform$ is then defined as $\faultform=\faultfreq(\faultkind)$.  Note that there is a natural order among our $4$ fault frequencies: a fault of kind $\faultkind$ that occurs permanently (frequency $\always$) is just a special case of the same fault $\faultkind$ occurring from some point onwards (frequency $\eventually \! \always$), which is in turn a special case of $\faultkind$ occurring infinitely often (frequency $ \always \! \eventually$), which is a special case of $\faultkind$ occurring at least once.  Thus, a test strategy that reveals a fault that occurs at least once (without knowing when) will also reveal a fault that occurs infinitely often, etc.  We say that $\eventually$ is the lowest and $\always$ is the highest fault frequency.  In our approach, we thus compute test strategies to detect faults at the lowest frequency for which a test strategy can be found.

\begin{algorithm}[tb]
\caption{\textsc{SyntLtlTest}: Synthesizes a universally complete test suite from an LTL specification for all outputs in $\outp$}
\label{alg:SyntLtlTest}
\begin{algorithmic}[1]
\ProcedureRet{SyntLtlTest}
             {\inp, \outp, \varphi, \faultkind}
             {A set $\suite$ of test strategies}
             \State $\suite:= \emptyset$
   \For{\textbf{each} $o_i \in \outp$} \label{alg:SyntLtlTest:0}
   \State $\suite := \suite\ \cup\ $\Call{SyntLtlIterate}{$\inp, \outp, \varphi, o_i, \faultkind,\emptyset$}; \label{alg:SyntLtlTest:1}
  \EndFor
  \State \textbf{return} $\suite$
\EndProcedure
\end{algorithmic}
\end{algorithm}

\begin{algorithm}[tb]
\caption{\textsc{SyntLtlIterate}: Synthesize an adaptive test strategy from an LTL specification with the lowest fault occurrence frequency}
\label{alg:SyntLtlIterate}
\begin{algorithmic}[1]
\ProcedureRet{SyntLtlIterate}
             {\inp, \outp, \varphi, o_i, \faultkind,\constr}
             {A singleton $\{ \strat \}$ with a test strategy~$\strat$ on success or $\emptyset$}
    \For{\textbf{each} $\faultfreq$ from $(\eventually,
                             \always \! \eventually,
                             \eventually \! \always,
                             \always )$
      in this order} \label{alg:SyntLtlIterate:1}
      \State $\strat := \synt_p\bigl(\outp \cup \{o_i'\}, \inp,
      \bigl(\varphi[o_i\leftarrow o_i'] \wedge \faultfreq(\faultkind)\bigr)
          \rightarrow
                \neg \varphi,
      \outp,\constr
      \bigr)$ \label{alg:SyntLtlIterate:2}
      \If{$\strat \neq \unreal$}
        \State \textbf{return} $\{\strat\}$; \label{alg:SyntLtlIterate:3}
      \EndIf
  \EndFor
  \State \textbf{return} $\emptyset$
\EndProcedure
\end{algorithmic}
\end{algorithm}

\paragraph{Algorithm}
The procedure \textsc{SyntLtlTest} in \refalg{SyntLtlTest} formalizes our approach using the procedure \textsc{SyntLtlIterate} in \refalg{SyntLtlIterate} as a helper.  The input consists of (1) the inputs $\inp$ of the SUT, (2) the outputs $\outp$ of the SUT, (3) an LTL specification $\varphi$ of the SUT, and (4) a fault kind $\faultkind$.  The result of \textsc{SyntLtlTest} is a test suite $\suite$.  The algorithm iterates over all outputs $o_i \in \outp$ (Line~\ref{alg:SyntLtlTest:0}) and invokes the procedure \textsc{SyntLTLIterate} (Line~\ref{alg:SyntLtlTest:1}).  The procedure \textsc{SyntLTLIterate} then iterates over the $4$ fault frequencies (Line~\ref{alg:SyntLtlIterate:1}), starting with the lowest one, and attempts to compute a strategy to reveal a fault (Line~\ref{alg:SyntLtlIterate:2}).  If such a strategy exists, it is returned to \refalg{SyntLtlTest} and added to $\suite$.  Otherwise, the procedures proceeds with the next higher fault frequency.

\paragraph{Sanity checks}
Note that our coverage goal in \refeq{complete_eq} is vacuously satisfied by any test suite if $\varphi$ or $\faultform$ is unrealizable.  The reason is that the test suite must reveal \emph{every} fault $F$ realizing $\faultform$ for \emph{every} system $\sys'$ realizing $\varphi$.  If there is no such fault or system, this is trivial.  As a sanity check, we thus test the (Mealy) realizability of $\varphi$ and $\always\faultkind$ before starting \refalg{SyntLtlTest} (because if $\always\faultkind$ is realizable, then so are $\eventually \always\faultkind$, $\always \eventually\faultkind$ and $\eventually\faultkind$). 

\paragraph{Handling unrealizability}
If, for some output, Line~\ref{alg:SyntLtlIterate:2} of \refalg{SyntLtlIterate} returns \unreal for the highest fault frequency $\faultfreq=\always$, we print a warning and suggest that the user examines these cases manually.  There are two possible reasons for unrealizability. First, due to limited observability, we do not find a test strategy although one exists (see~\refex{incompl}).  Second, no test strategy exists because there is some $\sys' \real \varphi[o_i\leftarrow o_i']$ and $F \real \faultform$ such that the composition $\sys = \sys' \concat F$ (see~\reffig{tg}) is correct, i.e., $\sys' \concat F \real \varphi$.  In other words, for some realization, adding the fault may result in an equivalent mutant in the sense that the specification is still satisfied.  For example, in case of a stuck-at-0 fault model, there may exist a realization of the specification that has the considered output $o_i\in\outp$ fixed to $\false$.  Such a high degree of underspecification is at least suspicious and may indicate unintended vacuities~\cite{BeerBER01} in the specification $\varphi$, which should be investigated manually.  If Proposition~\ref{th:worksfixed} or~\ref{th:worksstuck} applies, or if $\synt\bigl(\outp \cup \{o_i'\}, \inp, \bigl(\varphi[o_i\leftarrow o_i'] \wedge \always(\faultkind)\bigr) \rightarrow \neg \varphi,\constr \bigr)$ returns \unreal, we can be sure that the second reason applies.  Then, we can even compute additional diagnostic information in the form of two Mealy machines $\sys' \real \varphi[o_i\leftarrow o_i']$ and $F \real \faultform$ (by synthesizing some Mealy machine $S\real (\varphi[o_i\leftarrow o_i'] \wedge \always(\faultkind) \wedge \varphi)$ and splitting it into $\sys'$ and $F$ by stripping off different outputs).  The user can then try to find inputs for $\sys'\concat F$ such that the resulting trace violates the specification.  Failing to do so, the user will understand why no test strategy exists (see also~\cite{KonighoferHB13}).  For cases where the specification is as intended but no test strategy exists, we can follow the approach by Faella~\cite{Faella08,Faella09} to synthesize best-effort strategies that are not guaranteed to cause a specification violation but at least do not give up trying. But we leave this extension for future work.

\paragraph{Complexity}
Both $\synt_p(\outp, \inp, \psi, \outp',\constr)$ and $\synt(\outp, \inp, \psi,\constr)$ are 2EXPTIME complete in $|\psi|$~\cite{KupfermanV00}, so the execution time of \refalg{SyntLtlIterate}, and consequently also \refalg{SyntLtlTest}, are at most doubly exponential in $|\varphi| + |\faultkind|$.

\begin{theorem}\label{th:alg}
For a system with inputs $\inp$, outputs $\outp$, and LTL specification $\varphi$ over $\inp\cup \outp$, if the fault kind $\faultkind$ is of the form $\faultkind=\psi$ or $\faultkind = (o_i' \leftrightarrow \psi)$, where $\psi$ is an LTL formula over $\inp$ and $\outp$, $\textsc{SyntLtlTest}(\inp, \outp, \varphi, \faultkind)$ will return a universally complete test suite with respect to the fault model $\faultform=\always(\faultkind)$ if such a test suite exists.
\end{theorem}

\begin{proof}
Since $\always(\faultkind)$ implies $\faultfreq(\faultkind)$ for all $\faultfreq \in \{\eventually, \always \! \eventually, \eventually \! \always, \always\}$, Theorem~\ref{th:exists} and the guarantees of $\synt_p$ entail that the resulting test suite $\suite$ is universally complete with respect to $\faultform=\always(\faultkind)$ if $|\suite|=|\outp|$, i.e., if \textsc{SyntLtlTest} found a strategy for every output.  It remains to be shown that $|\suite|=|\outp|$ for $\faultkind=\psi$ or $\faultkind = (o_i' \leftrightarrow \psi)$ if a universally complete test suite for $\faultform=\always(\faultkind)$ exists: either Proposition~\ref{th:worksfixed} or Proposition~\ref{th:worksstuck} states that \refeq{theo_eq} holds with $\faultform=\always(\faultkind)$.  Thus, $\synt_p$ cannot return $\unreal$ in \textsc{SyntLtlIterate} with $\faultfreq = \always$, so $|\suite|$ must be equal to $|\outp|$ in this case.
\end{proof}

Theorem~\ref{th:alg} states that \textsc{SyntLtlTest} is not only sound but also complete for many interesting fault models such as stuck-at faults or permanent bit-flips.  For $\faultkind=\psi$, Theorem~\ref{th:alg} can even be strengthened to hold for all $\faultform=\faultfreq(\faultkind)$ with $\faultfreq \in \{\eventually, \always \! \eventually, \eventually \! \always, \always\}$.

\subsection{Extensions and Variants}
\label{sec:ext}

A test suite computed by \textsc{SyntLtlTest} for specification $\varphi$ and fault model $\faultform$ is universally complete and detects all faults with respect to $\varphi$ and $\faultform$ independent of the implementation and the concrete fault manifestation if the fault manifests at one of the observable outputs as illustrated in \reffig{tg}.

In this section, we discuss some alternatives and extensions of our approach to improve fault coverage and performance.

\paragraph{User-specified fault frequencies}
Besides the four fault frequencies ($\always$, $\eventually \!\always$, $\always \!\eventually$, and $\eventually$), other fault frequencies (with different precedences) may be of interest, e.g., if a specific time step is of special interest.  \refalg{SyntLtlIterate} supports full LTL and thus the procedure can be extended by replacing Line~\ref{alg:SyntLtlIterate:1} by ``\textbf{for each} $\faultfreq$ from $\faultfreqset$ in this order'', where $\faultfreqset$ is an additional parameter provided by the user.

\paragraph{Faults at inputs}
In the fault model in the previous section,  we only consider faults at the outputs.  However, considering SUTs that behave as if they would have read a faulty input is possible as well (by changing Line~\ref{alg:SyntLtlTest:0} in \refalg{SyntLtlTest} to ``\textbf{for each} $o\in \inp \cup \outp$ \textbf{do}'').  

\paragraph{Multiple faults}
Faults that occur simultaneously at multiple (inputs or) outputs $\{o_1, \ldots, o_k\}\subseteq \outp$ can be considered by computing a test strategy
\begin{align*}
\strat := \synt_p\Big(\outp \cup \{o_1', \ldots, o_k'\}, \inp, (\varphi[o_1\leftarrow o_1', \ldots, o_k\leftarrow o_k'] \wedge \bigwedge_{i=1}^{k} \faultform_i) \rightarrow \neg \varphi, \outp,\constr\Big),
\end{align*}
where the fault model $\faultform_i$ can be different for different outputs $o_i \in \{o_1, \ldots, o_k\}$.

\paragraph{Faults within a SUT}
If a fault manifests in a \emph{conditional fault} in a system implementation, a universally complete $\suite$ may not be able to uncover the fault (see \refex{morestrats}).

\begin{figure}[t]%
  \centering
  \subfloat{\begin{tikzpicture}[>=latex,->,auto,initial text={},initial distance=3mm]
\node[state,initial,inner sep=0] at  (0,0)    (S0) {$\mathsf{i}$};
\node[state,inner sep=0]         at  (1.5,0)  (S1) {$\neg \mathsf{i}$};
\node[]                          at  (0,-1.3) (SENTINEL) {};
\path
(S0) edge [bend left]  node[xshift=0mm,yshift=0mm] {$*$} (S1)
(S1) edge [bend left]  node[xshift=0mm,yshift=0mm] {$*$} (S0)
;
\end{tikzpicture}}%
  \qquad\qquad
  \subfloat{\begin{tikzpicture}[>=latex,->,auto,initial text={},initial distance=3mm]
\node[state,initial,inner sep=0]  at  (0,0)       (S0) {$$};
\node[state,inner sep=0]          at  (1.5,0.5)     (S1) {$$};
\node[state,inner sep=0]          at  (1.5,-0.5)     (S2) {$$};

\path
(S0) edge [ultra thick,sloped] node[xshift=3mm,yshift=0mm] {$i / o$} (S1);
\path
(S0) edge [sloped]  node[xshift=-5mm,yshift=-5mm] {$\neg i / \neg o$} (S2);
\path
(S1) edge [loop above]  node[xshift=5mm,yshift=-5mm] {$* / o$} (S1);
\path
(S2) edge [loop below]  node[xshift=5mm,yshift=5mm] {$* / \neg o$} (S2);
\end{tikzpicture}}%
  \caption{Test strategy $\strat_6$ and a faulty system implementation of the specification $\varphi = \always((i \leftrightarrow \nextt(\neg  i)) \rightarrow \nextt(o))$.}%
  \label{fig:strategyGswitchio}%
\end{figure}
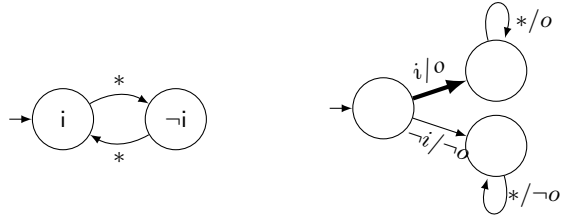

\begin{example}\label{ex:morestrats}
Consider a system with input $\inp=\{i\}$, output $\outp=\{o\}$, and specification $\varphi = \always((i \leftrightarrow \nextt(\neg  i)) \rightarrow \nextt(o))$.  The specification enforces $o$ to be set to $\true$ whenever input $i$ alternates between $\true$ and $\false$ in consecutive time steps.  Consider a stuck-at-$0$ fault $\faultform = \always \! \eventually \neg o$ at the output $o$.  The test suite $\suite = \{\strat_6\}$ with the test strategy $\strat_6$ illustrated in \reffig{strategyGswitchio} (on the left) is universally complete with respect to $\faultform$.  The test strategy $\strat_6$ flips input $i$ in every time step and thus forces the system to set $o = \true$ in the second time step.  Now consider the concrete and faulty system implementation in \reffig{strategyGswitchio} (on the right) of $\varphi$.  The test strategy $\strat_6$, when executed, first follows the bold edge and then remains forever in the same state.  As a consequence, the fault in the system implementation, i.e., $o$ stuck-at-$0$, is not uncovered.  To uncover the fault, $i$ has to be set to $\false$ in the initial state.
\end{example}

Faults within a system implementation can be considered by computing more than one test strategy for a given test objective.  We extend \refalg{SyntLtlTest} to generate a bounded number~$b$ of test strategies by setting $\constr$~=~$\suite$ in Line~\ref{alg:SyntLtlTest:1} and enclosing the line by a \textbf{while}-loop that uses an additional integer variable~$c$ to count the number of test strategies generated per output $o_i$.
The \textbf{while}-loop terminates if no new test strategy could be generated or if $c$ becomes equal to $b$.
Note that this approach is correct in the sense that all computed test strategies are universally complete  with respect to the fault model $\faultfreq(\faultkind)$; however, in many cases it is more efficient to determine the lowest fault frequency first in Line~\ref{alg:SyntLtlTest:1} of Alg.~\ref{alg:SyntLtlIterate} and then generate multiple test strategies with the same (or higher) frequency by enclosing Line~\ref{alg:SyntLtlIterate:2} with the \textbf{while}-loop.

\paragraph{Test strategy generalization}
A synthesis procedure usually assigns concrete values to all variables in every state of the generated test strategy.  In many cases, however, not all assignments are necessary to enforce a test objective (see \refex{generalizeStratl}).

\begin{figure}[t]%
  \centering
  \subfloat{\begin{tikzpicture}[>=latex,->,auto,initial text={},initial distance=3mm]
\node[state,initial,inner sep=0]  at (0,0) (S0) {\begin{minipage}{1cm}\centering$\mathsf{r_1}$ \\$\neg \mathsf{r_2}$\end{minipage}};

\path
(S0) edge [loop right] node[xshift=-0mm,yshift=0mm]    {$*$}      (S0)
;
\end{tikzpicture}}%
  \qquad\qquad
  \subfloat{\begin{tikzpicture}[>=latex,->,auto,initial text={},initial distance=3mm]
\node[state,initial,inner sep=0]  at (0,0) (S0) {\begin{minipage}{1cm}\centering$\mathsf{r_1}$ \\$\mathsf{r_2}$\end{minipage}};

\path
(S0) edge [loop right] node[xshift=-0mm,yshift=0mm]    {$*$}      (S0)
;
\end{tikzpicture}}%
  \qquad\qquad
  \subfloat{\begin{tikzpicture}[>=latex,->,auto,initial text={},initial distance=3mm]
\node[state,initial,inner sep=0]  at (0,0) (S0) {\begin{minipage}{1cm}\centering$\mathsf{r_1}$ \end{minipage}};

\path
(S0) edge [loop right] node[xshift=-0mm,yshift=0mm]    {$*$}      (S0)
;
\end{tikzpicture}}%
  \caption{Test strategy $\strat_7$ on the left, $\strat_8$ in the middle and $\strat_9$ on the right.}%
  \label{fig:generalized strategies}%
\end{figure}
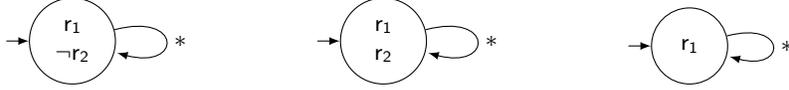

\begin{example}\label{ex:generalizeStratl}
Consider a system with inputs $\inp=\{r_1, r_2\}$ and outputs $\outp=\{g_1, g_2\}$, which implements the specification of a two-input arbiter
$\varphi = \always(r_1 \rightarrow \eventually g_1) \land \always(r_2 \rightarrow \eventually g_2) \land \always( \neg  g_1 \lor \neg  g_2)$, i.e., every request $r_i$ shall eventually be granted by setting $g_i$ to $\true$ and there shall never be two grants at the same time.  A valid test strategy $\strat_7$ that tests for a stuck-at-0 fault of signal $g_1$ from some point in time onwards may simply set $r_1=\true$ and $r_2=\false$ all the time (see \reffig{generalized strategies}). This forces the system in every time step to eventually grant this one request by setting $g_1 = \true$. Another valid test strategy $\strat_8$  sets $r_1=\true$ and $r_2=\true$ all the time  (see \reffig{generalized strategies}). Now the system has to grant both requests eventually. Both $\strat_7$ and $\strat_8$ test for the defined stuck-at-0 fault of signal $g_1$ from some point in time onwards but will likely execute different paths in the SUT. Thus, considering the more general strategy $\strat_9$  (see \reffig{generalized strategies}) that sets $r_1=\true$ all the time but puts no restrictions on the value of $r_2$, allows the tester to evaluate different paths in the SUT while still testing for the defined fault class.
\end{example}

\begin{algorithm}[tb]
\caption{\textsc{Generalize}: Generalize a test strategy.}
\label{alg:GeneralizeStrat}
\begin{algorithmic}[1]
\ProcedureRet{Generalize}
             {\inp, \outp, \varphi, o_i, \faultfreq, \faultkind, \strat}
             {A generalization of $\strat$}
  \For{\textbf{each} $q_i \in T$} \label{alg:GeneralizeStrat:1}
    \For{\textbf{each} $x_i \in \inalp$}  \label{alg:GeneralizeStrat:2}
      \State $\strat' := \text{remove assignment to } x_i \text{ from state } q_i \text{ in } \strat$ \label{alg:GeneralizeStrat:3} 
      \If{$\modelcheck l(T^\prime,  \bigl( \varphi[o_i \leftarrow o_i^\prime] \wedge \faultfreq(\faultkind)\bigr)  \rightarrow  \neg \varphi)$} \label{alg:GeneralizeStrat:4} 
      \State $\strat := \strat'$
      \EndIf
    \EndFor
  \EndFor
  \State \textbf{return} $\strat$
\EndProcedure
\end{algorithmic}
\end{algorithm}

The procedure in \refalg{GeneralizeStrat} generalizes a given test strategy $\strat$ by systematically removing variable assignments from states and employing a modelchecking procedure to ensure that the generalized test strategy still enforces the same test objective.  The procedure loops in Line~\ref{alg:GeneralizeStrat:1} over all states of $\strat$ and in Line~\ref{alg:GeneralizeStrat:2} over all inputs.  In Line~\ref{alg:GeneralizeStrat:3} the assignment to the input $x_i$ in a state is removed such that the corresponding variable becomes non-deterministic. If the resulting test strategy still enforce the test objective, then $\strat$ is replaced by its generalization.  Otherwise, the change is reverted.  \refalg{GeneralizeStrat} is integrated into \refalg{SyntLtlIterate} and applied in Line~\ref{alg:SyntLtlIterate:3} to generalize each generated test strategy.

Note that generalizing a test strategy is a a special way of computing multiple concrete test strategies, which was discussed in the previous section. However, generalization may fail when computing multiple strategies succeeds (by following different paths).

\paragraph{Optimization for full observability}
If we restrict our perspective to the case with no partial information, i.e., all signals are fully observable, we can employ the optimization discussed in Proposition~\ref{th:worksstuck} to improve the performance of test strategy generation.  In Line~\ref{alg:SyntLtlIterate:2} of \refalg{SyntLtlIterate} we drop a part of the assumption and simplify the synthesis step to $\strat_i := \synt\bigl(\outp, \inp, \faultfreq(\faultkind) \rightarrow \neg  \varphi, \constr \bigr)$ for cases in which $\faultkind$ does not refer to a hidden signal $o_i'$.  Also, for a fault model $\faultform$ that describes a fault of kind $\faultkind = (o_i' \leftrightarrow \psi)$, where $\psi$ is an LTL formula over $\inp$ and $\outp$, we can drop the part of the assumption according to Proposition~\ref{th:worksfixed} if $\faultfreq=\always$. This simplifies Line~\ref{alg:SyntLtlIterate:2} of \refalg{SyntLtlIterate} to $\strat_i := \synt\bigl(\outp, \inp, \varphi[o_i\leftarrow \psi] \rightarrow \neg  \varphi,\constr \bigr)$. These simplifications, moreover, no longer require a synthesis procedure with partial information and thus, a larger set of synthesis tools is supported.

\paragraph{Mutating the specification}
We can also synthesize adaptive test strategies that would uncover bugs where the SUT implements a mutated (i.e., slightly modified) specification $\varphi'$ instead of $\varphi$ by calling $\strat := \synt(\outp, \inp, \varphi' \rightarrow \neg \varphi,\constr)$.  The implication requires the original specification $\varphi$ to be violated under the assumption that the mutated specification $\varphi'$ has been implemented in the SUT.  This variant does not require partial information synthesis.

\paragraph{Other specification formalisms}
We worked out our approach for LTL, but it works for other languages if (1) the language is closed under Boolean connectives $(\wedge, \neg)$, (2) the desired fault models are expressible, and (3) a synthesis procedure (with partial information) is available.  These prerequisites do not only apply to many temporal logics but also to various kinds of automata over infinite words.

\section{Case Study}
\label{sec:experimentals}

To evaluate our approach, we apply it in a case study on a real component of a satellite that is currently under development. We first present the system under test and specify a version of the respective component in LTL. Using this specification, we compute a set of test strategies and evaluate the test suite on a real implementation. Additional case studies can be found in~\cite{BloemKPR16}.

\subsection{\eucropis FDIR Specification}
\label{sec:eucropisfdir}

An important task of each space and satellite system is to maintain its health state and react on failure.  In modern space systems this task is encapsulated in the \emph{Fault Detection, Isolation, and Recovery} (FDIR) component, which collects the information from all relevant sensors and on-board computers, analyzes and assess the data in terms of correctness and health, and initiates recovery actions if necessary.  The FDIR component is organized hierarchically in multiple levels~\cite{TipaldiB15} with the overall objective of maximizing the system life-time and correct operation.

In this section, we focus on system-level FDIR and present the high-level abstraction of a part of the FDIR mechanisms used in the \eucropis satellite mission as a case-study for adaptive test strategy generation.  On the system-level, the FDIR mechanism deals with coarse-granular anomalies of the system behavior like erroneous sensor data or impossible combinations of signals.  Likewise the recovery actions are limited to restarting certain sub-systems, switching between redundant sub-systems if available, or switching into the satellite's safe mode.  The FDIR component is highly safety- and mission-critical; if recovery on this level fails, in many cases the mission has to be considered lost.

\begin{figure}[t]
  \centering
  \includegraphics[width=1.0\textwidth]{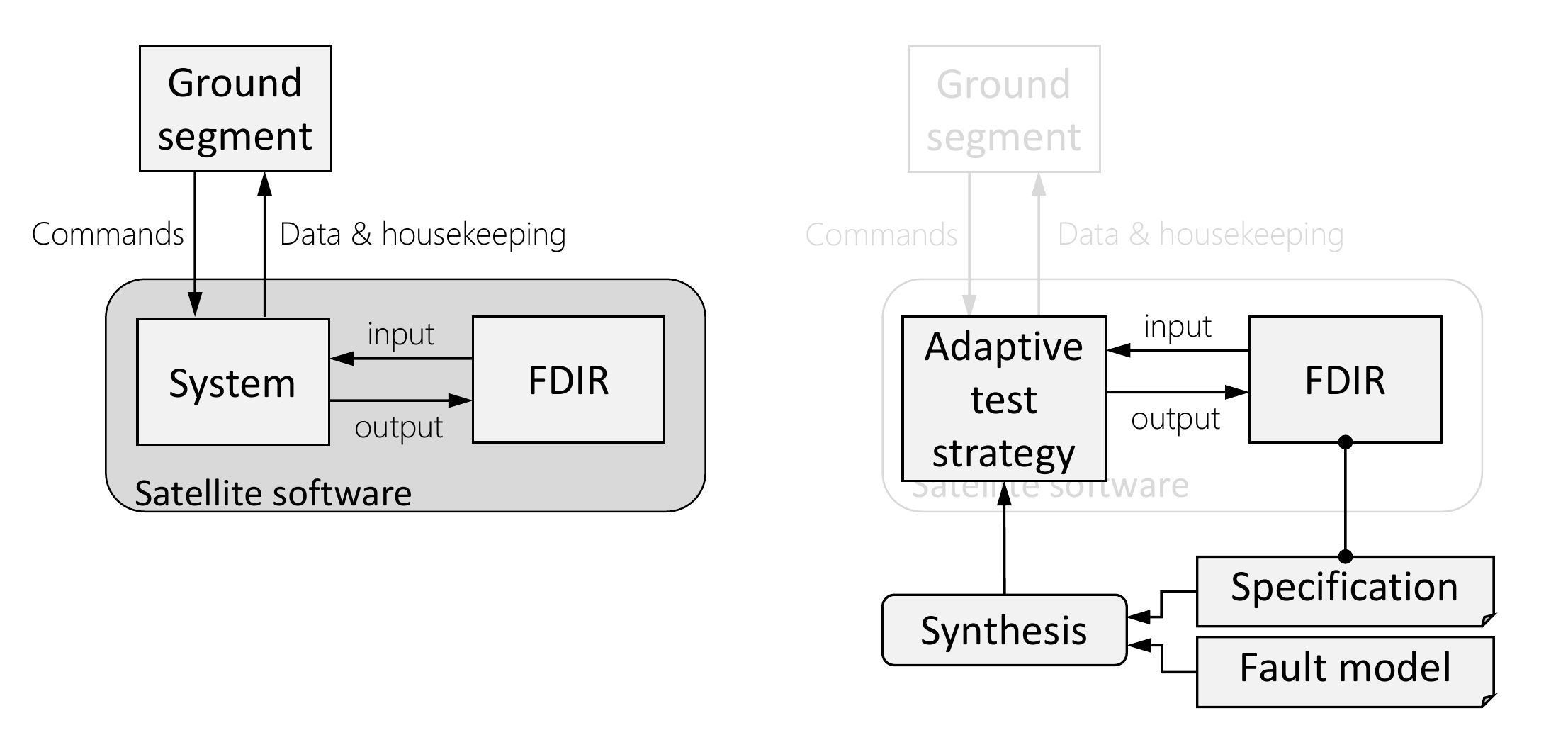}
  \caption{FDIR in practice (left) and the intended test setup (right).}
  \label{fig:fdirinpractice}
\end{figure}

\begin{figure}[t]
  \centering
  \includegraphics[width=1.0\textwidth]{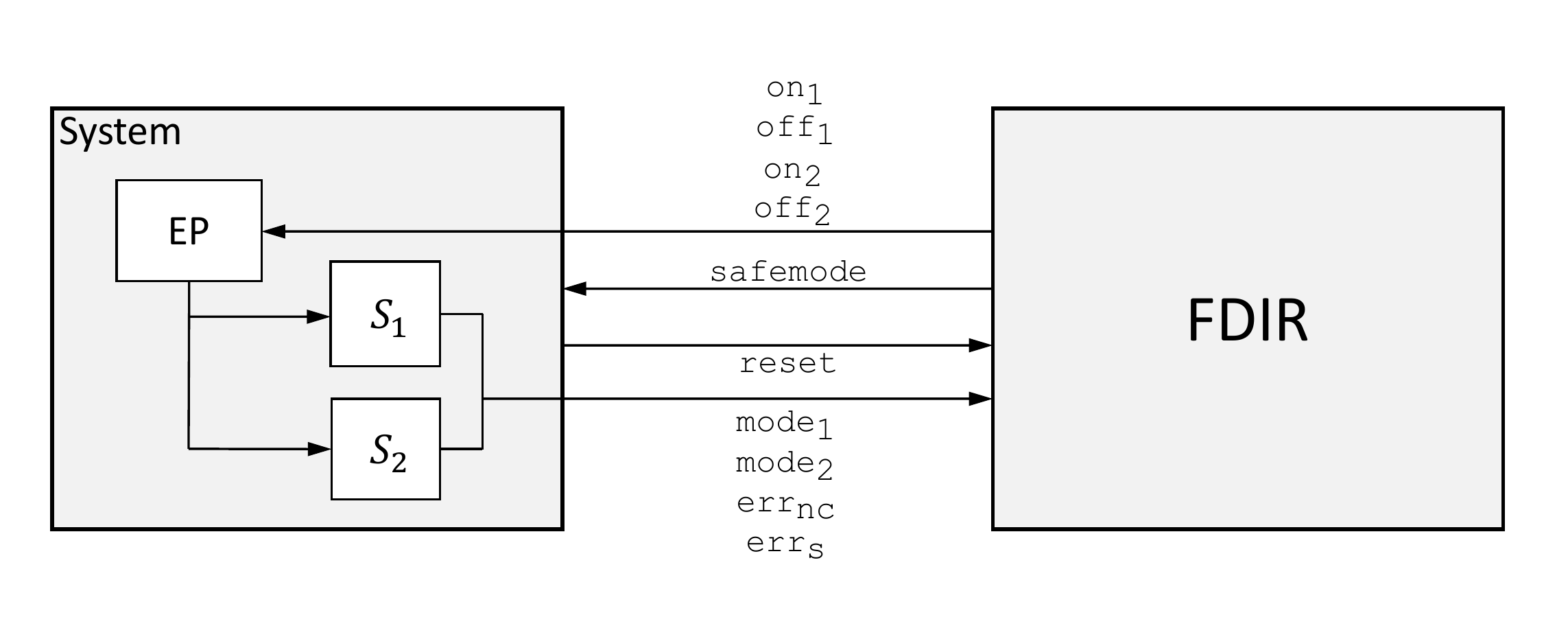}
  \caption{High-level overview of the satellite software of \eucropis.}
  \label{fig:fdir overview}
\end{figure}

\paragraph{\eucropis FDIR}
In  \reffig{fdirinpractice} we illustrate where the FDIR component for the magnetic torquers of the \eucropis on-board computing system is placed in practice and in  \reffig{fdir overview}, we give a high-level overview of the FDIR component and its environment.   The FDIR component regularly obtains housekeeping information from two redundantly-designed control units, $S_1$ and $S_2$, which control the magnetic torquers of the satellite, and interacts with them via the electronic power system, EP.  The control units $S_1$ and $S_2$ have the same functionality, but only one of them is active at any time.  The other control unit serves as a backup that can be activated if necessary.  The FDIR component signals the activation (or deactivation) of a control unit to the EP which regulates the power supply.

We distinguish two types of errors, called \emph{non-critical error} and \emph{severe error}, signaled to the FDIR component via housekeeping information.  In case of a non-critical error, two recovery actions are allowed.  Either the erroneous control unit is disabled for a short time and enabled afterwards again or the erroneous control unit is disabled and the redundant control unit is activated to take over its task.  In case of the severe error, however, only the latter recovery action is allowed, i.e., the erroneous control unit has to be disabled and the redundant control unit has to be activated.  If this happens more than once and the redundant control unit as well shows erroneous behavior, the FDIR component initiates a switch of the satellite mode into safe mode.  The safe mode is a fall-back satellite mode designed to give the operators on ground the maximum amount of time to analyze and fix the problem.  It is only invoked once a problem cannot be solved on-board and requires input from the operators to restore nominal operations.

\begin{table}[t]
  \centering
  \caption{Descriptions of inputs and outputs of the FDIR component.}
  \label{tab:signals}
  \begin{tabularx}{\textwidth}{p{4cm} p{7cm}}
    \toprule
    \multicolumn{1}{c}{\bf Boolean variable} &
    \multicolumn{1}{c}{\bf Description} \\
    \midrule
    \nommode     & $\true$ iff $S_1$ is activated \\
    \redmode     & $\true$ iff $S_2$ is activated \\
    \normerr     & $\true$ iff a non-critical error is signaled by $S_1$ or $S_2$ \\
    \criterr     & $\true$ iff a severe error is signaled by $S_1$ or $S_2$ \\
    \reset       & $\true$ iff the FDIR component is reset \\
    \midrule
    \nomon       & $\true$ iff $S_1$ shall be switched on \\
    \nomoff      & $\true$ iff $S_1$ shall be switched off \\
    \redon       & $\true$ iff $S_2$ shall be switched on \\
    \redoff      & $\true$ iff $S_2$ shall be switched off \\
    \safemode    & $\true$ iff the FDIR component initiates the safemode of the satellite \\
    \midrule
    \lastupisnom & $\true$ if the last active system was $S_1$ and $\false$ if the last active system was $S_2$ \\
    \allowswitch & $\true$ iff a switch of $S_1$ to $S_2$ or $S_2$ to $S_1$ is allowed \\
    \bottomrule
  \end{tabularx}
\end{table}

\paragraph{LTL specification}
We model the specification of the FDIR component in LTL.  Let $\inp_{FDIR}$ = \{$\nommode$, $\redmode$, $\normerr$, $\criterr$, $\reset$\} and $\outp_{FDIR}$ = \{$\nomon$, $\nomoff$, $\redon$, $\redoff$, $\safemode$\} be the Boolean variables corresponding to the input signals and the output signals of the FDIR component, respectively.

These Boolean variables are abstractions of the real hardware/software implementation.  The values of the Boolean variables are automatically extracted from the housekeeping information which is periodically collected from EP ($\nommode$, $\redmode$) and $S_1$ or $S_2$ ($\normerr$, $\criterr$).  The two error variables encompass multiple error conditions (e.g. communication timeouts, invalid responses, electrical errors like over-current or under-voltage, etc.) which are detected by the sub-system.  The $\reset$ variable corresponds to a telecommand sent from ground to the FDIR component.  For the output direction the values of the variables are used to generate commands which are sent to the EP or the satellite mode handling component.  Additionally, we use the auxiliary Boolean variables $\outp^\prime$~=~\{$\lastupisnom$, $\allowswitch$\} to model state information on specification level which does not correspond to any real signals in the system.  These auxiliary variables serve as unobservable outputs of the FDIR component.  In \reftab{signals}, we summarize the Boolean variables involved in the specification and their meaning.

The complete LTL specification of the FDIR component consists of the assumptions A1-A6 and the guarantees G1-G13.  All properties are listed in \reftab{fdir}, expressing the following intentions:

\begin{itemize}
\item[A1] Whenever both systems are off, then there is no running system that can have an error. Thus, the error signals have to be low as well.
\item[A2] The error signals are mutual exclusive. If the environment enforces a reset then both error signals have to be low, because we assume that ground control has taken care of the errors.
\item[A3] After a reset enforced by the environment, one of the two systems has to be running and the other has to be off.
\item[A4] Whenever the FDIR component sends $\nomon$, we assume that in the next time step system number one is running ($\nommode$) and the state of the second system ($\redmode$) does not change. The same assumption applies analogously for $\redon$.
\item[A5] Whenever the FDIR component sends $\nomoff$, we assume that in the next time step system number one is off ($\neg \nommode$) and the state of the second system ($\redmode$) does not change. The same assumption applies analogously for $\redoff$.
\item[A6] We assume that the environment, more specifically the electronic power unit, is not immediately free to change the state of the systems when there is no message from the FDIR component. It has to wait for one more time step (with no messages of the FDIR component).
\item[G1]This guarantee stores which system was last activated by the FDIR component.
\item[G2] We require the signals $\nomon$, $\nomoff$, $\redon$ and $\redoff$ to be mutually exclusively set to high.
\item[G3] Whenever both systems are off, then the FDIR component eventually requests to switch on one of the systems ($\nomon$, $\redon$) or activates $\safemode$ or observes a $\reset$.
\item[G4] We restrict the FDIR component to not enter $\safemode$ as long as the component can switch to the backup system.
\item[G5] The FDIR component must not request to switch on one of the systems ($\nomon$, $\redon$) as long as one of the systems is running.
\item[G6] Whenever the FDIR component is not allowed anymore to switch to the backup system, then it must not request to switch the backup system on.
\item[G7] Once the FDIR component switches to the backup system it is not allowed anymore to switch again (unless the environment performs a reset, see G9).
\item[G8] As long as the FDIR component only restarts the same system it is still allowed to switch in the future.
\item[G9] A $\reset$ by the environment allows the FDIR component again to switch to the backup system if required.
\item[G10] Whenever the FDIR component is in $\safemode$ it must not request to switch-on one of the systems  ($\nomon$,$\redon$).
\item[G11] Once a switch is not allowed anymore and the environment does not perform a reset, then the switch is also not allowed in the next time step.
\item[G12] Whenever the FDIR component observes a server error ($\criterr$), it must eventually switch to the backup system or activate $\safemode$ unless the environment performs a $\reset$ or the error disappears by itself (without restarting the system).
\item[G13] Whenever the FDIR component observes a non-critical error ($\normerr$), it must eventually switch to the backup system or activate $\safemode$ or the error disappears (restarting the currently running system is allowed).
\end{itemize}

\begin{table}[p]
  \centering
  \caption{Temporal  specification of system-level FDIR component in LTL.}
  \label{tab:fdir}
  \begin{tabularx}{\textwidth}{Z p{10.6cm}}
    \toprule
    \multicolumn{2}{c}{\bf Assumptions A1--A6} \\
    \midrule
    A1 & $\always (\neg  \redmode \land \neg  \nommode \rightarrow \neg  \normerr \land \neg  \criterr )$ \\
    \midrule
    A2 & $\always (\neg \normerr \vee \neg \criterr) \wedge \always (\reset \rightarrow \neg \normerr \land \neg \criterr)$ \\
    \midrule
    A3 & $\always (\reset \rightarrow \nextt (\redmode \oplus \nommode))$ \\
    \midrule
    A4 & $\begin{aligned}[t]
          \always ( & \neg \nommode \land \nomon \land \neg \nomoff \land \neg \redon \land \neg \redoff \land \neg \reset \land \neg \safemode \rightarrow \\
                    & \nextt (\nommode) \land (\redmode \leftrightarrow \nextt (\redmode)))
          \end{aligned}$ \\
       & $\begin{aligned}[t]
          \always ( & \neg \redmode \land \neg \nomon \land \neg \nomoff \land \redon \land \neg \redoff \land \neg \reset \land \neg \safemode \rightarrow \\
                    & \nextt (\redmode) \land (\nommode \leftrightarrow  \nextt (\nommode)))
          \end{aligned}$ \\
    \midrule
    A5 & $\begin{aligned}[t]
          \always ( & \nommode \land \neg \nomon \land \nomoff \land \neg \redon \land \neg \redoff \land \neg \reset \land \neg \safemode \rightarrow \\
                    & \nextt (\neg \nommode) \land (\redmode \leftrightarrow \nextt(\redmode)) )
          \end{aligned}$ \\
       & $\begin{aligned}[t]
          \always ( &\redmode \land \neg \nomon \land \neg \nomoff \land \neg \redon \land \redoff \land \neg \reset \land \neg \safemode \rightarrow \\
                    & \nextt (\neg \redmode) \land (\nommode \leftrightarrow \nextt(\nommode)) )
          \end{aligned}$ \\
    \midrule
    A6 & $\begin{aligned}[t]
          \always ( &(\neg  (\neg \redon \land \neg \nomoff \land \neg \nomon \land \neg \redoff) \land \nextt(\neg \redon \land \neg \nomoff \land \neg \nomon \land \neg \redoff) \land \\
                    &(\neg \reset \land \nextt(\neg \reset) \land \neg \safemode \land \nextt(\neg \safemode)) \rightarrow \\
                    &\nextt ((\redmode \leftrightarrow \nextt(\redmode)) \land (\nommode \leftrightarrow \nextt(\nommode) ) )
          \end{aligned}$ \\
    \midrule
    \multicolumn{2}{c}{\bf Guarantees G1--G13} \\
    \midrule
    G1 & $\always((\nomon \land \neg \redon) \rightarrow (\nextt(\lastupisnom)))$ \\
       & $\always((\neg \nomon \land \redon) \rightarrow (\nextt(\neg \lastupisnom)))$ \\
       & $\always( (\neg \nomon \land \neg \redon) \rightarrow (\lastupisnom \leftrightarrow \nextt(\lastupisnom)))$ \\
    \midrule
    G2 & $\always(\nomon \rightarrow \neg \nomoff \land \neg \redon \land \neg \redoff)$ \\
       & $\always(\nomoff \rightarrow \neg \nomon \land \neg \redon \land \neg \redoff)$ \\
       & $\always(\redon \rightarrow \neg \nomon \land \neg \nomoff \land \neg \redoff)$ \\
       & $\always(\redoff \rightarrow \neg \nomon \land \neg \redon \land \neg \nomoff)$ \\
    \midrule
    G3 &
    $\always (\neg \redmode \land \neg \nommode \rightarrow \eventually(\reset \lor \redon \lor \nomon \lor \safemode))$ \\
    \midrule
    G4 & $\always(\allowswitch \rightarrow \neg \safemode)$ \\
    \midrule
    G5 & $\always((\redmode \lor \nommode) \rightarrow \neg \nomon \land \neg \redon)$ \\
    \midrule
    G6 & $\always(\neg \allowswitch \land \lastupisnom \rightarrow \neg \redon)$ \\
     & $\always(\neg \allowswitch \land \neg \lastupisnom \rightarrow \neg \nomon)$ \\
    \midrule
    G7 & $\always (\neg \reset \land \allowswitch \land \lastupisnom \land \redon \rightarrow \nextt (\neg \allowswitch) )$ \\
       & $\always (\neg \reset \land \allowswitch \land \neg \lastupisnom \land \nomon \rightarrow \nextt (\neg \allowswitch) )$ \\
    \midrule
    G8 & $\always ((\allowswitch \land  \neg (((\lastupisnom \land \redon) \lor (\neg \lastupisnom \land \nomon)))) \rightarrow \nextt (\allowswitch))$ \\
    \midrule
    G9 & $\always (\reset \rightarrow \nextt (\allowswitch))$ \\
    \midrule
    G10 & $\always(\safemode \rightarrow (\neg \nomon \land \neg \redon))$ \\
    \midrule
    G11 & $\always( \neg \allowswitch \land  \neg \reset \rightarrow \nextt (\neg \allowswitch))$ \\
    \midrule
    G12 & $\begin{aligned}[t]
           \always( &(\criterr \land \nommode\land \neg \reset) \rightarrow \\
                    &\eventually(\reset \lor \safemode \lor \redmode \lor (\nommode \until (\nommode \land \neg \criterr))))
          \end{aligned}$ \\
        & $\begin{aligned}[t]
           \always( &(\criterr \land \redmode\land \neg \reset) \rightarrow \\
                    &\eventually(\reset \lor \safemode \lor \nommode \lor (\redmode \until (\redmode \land \neg \criterr))))
          \end{aligned}$ \\
    \midrule
    G13 & $\always((\normerr \land  \nommode \land \neg \reset) \rightarrow \eventually(\reset \lor \safemode \lor \redmode \lor (\nommode \land \neg \normerr)))$ \\
        & $\always((\normerr \land  \redmode \land \neg \reset) \rightarrow \eventually(\reset \lor \safemode \lor \nommode \lor (\redmode \land \neg \normerr)))$ \\
    \bottomrule
  \end{tabularx}
\end{table}

\subsection{Experimental Results}

The test strategy computation from the specification is independent of the implementation.  Thus, we first present the experimental results of the strategies derived from the LTL specification of the FDIR component given in ~\reftab{fdir}, then we execute and evaluate the computed strategies on the implementation of the specification in the system of the \eucropis satellite.

\subsubsection{Test strategy computation}
\label{sec:stratcomp}

\paragraph{Experimental setting}
All experiments for computing the test strategies are conducted in a virtual machine with a 64 bit Linux system using a single core of an Intel i5 CPU running at $2.60$\,GHz.  We use the synthesis procedure \party~\cite{KhalimovJB13} as black-box, which implements SMT-based bounded synthesis for full LTL and, thus, we call our tool \partystrategy.

\begin{table}[t]
  \centering
  \caption{Results for the FDIR specification.  The suffix ``k'' multiplies by $10^3$.}
  \label{tab:fdir:stratcalc}
  \newcommand{\GF}{$\always\!\eventually$}
  \newcommand{\FG}{$\eventually\! \always$}
  \begin{tabularx}{\textwidth}{p{1cm}p{.8cm}p{1.8cm}ZZZZ}
    \toprule
    \multicolumn{2}{c}{\bf Fault} &
    \multicolumn{1}{c}{\bf $o_i$} &
    \multicolumn{1}{Z}{\bf $\faultfreq$} &
    \multicolumn{1}{Z}{\bf $|\strat|$} &
    \multicolumn{1}{Z}{\bf Time} &
    \multicolumn{1}{Z}{\bf Peak Memory} \\
    &&&&&
    \multicolumn{1}{c}{[s]} &
    \multicolumn{1}{c}{[MB]} \\
    \midrule
    \multirow{3}{*}{S-a-0}
    &
    &  \nomon    & \FG & 4 &   1.2k & 400 \\
    && \nomoff   & \FG & 3 &    517 & 396 \\
    && \safemode & \FG & 4 &    934 & 324 \\
    \midrule
    \multirow{3}{*}{S-a-1}
    &
    &  \nomon    & \GF & 4 &    438 & 222 \\
    && \nomoff   & \FG & 4 &    753 & 378 \\
    && \safemode & \GF & 3 &    169 & 192 \\
    \midrule
    \multirow{3}{*}{Bit-Flip}
    &
    &  \nomon    & \GF & 4 &   26k & 3.6k \\
    && \nomoff   & \FG & 4 & 98.9k & 4.3k \\
    && \safemode & \GF & 3 & 13.1k & 4.3k \\
    \bottomrule
  \end{tabularx}
\end{table}

\paragraph{Test strategy computation}
From the previously described LTL specification, we compute test strategies for the outputs \nomon, \nomoff and \safemode of the FDIR component considering the fault models stuck-at-0, stuck-at-1, and bit-flip with the lowest possible fault frequencies.  These are general fault assumptions and cover faults where the specification is violated with this signal being high (stuck-at-1), faults where the specification is violated with this signal being low (stuck-at-0) and faults where the specification is violated with this signal having the wrong polarity (bit-flip).  We do not synthesize test strategies for the outputs \redon and \redoff because they behave identical to \nomon and \nomoff, respectively, if the role of $S_1$ and $S_2$ are mutually interchanged.  For synthesizing test strategies, both, the bound for the maximal number of states of a test strategy and the bound for the maximal number of test strategies, are set to four. We chose the bound to be four, because for this bound there exist strategies for all our chosen fault models and output signals. The size for the maximum number of strategies per variable and fault model is set arbitrarily to four and could also be set to a different value.

In \reftab{fdir:stratcalc}, we list the time and memory consumption for synthesizing the test strategies with our synthesis tool \partystrategy. The more freedom there is for implementations of the specification, the harder it becomes to compute a strategy. The search for strategies that are capable of detecting a bit-flip is the most difficult one as we cannot make use of our optimization for full observability of the output signals.  For all signals with a stuck-at-0 fault and for the \nomoff signal with one of the other two faults we are able to derive test strategies that can detect the fault if it is permanent from some point onwards. For the signals \nomon and \safemode we are able to derive strategies for stuck-at-1 faults and bit-flips also at a lower frequency, i.e., we can detect those faults also if they occur at least infinitely often.

\paragraph{Illustration of a computed strategy}
We illustrate and explain one derived strategy in detail. The strategy derived for the signal \safemode being stuck-at-0 computed with \partystrategy consists of four states. \reffig{FDIRsatmodesafeG0} illustrates the strategy. In the first state (state 0) we have the first system running (\nommode) and set the \normerr flag, i.e., we raise a non critical error that requires the component to restart until the error is gone or to switch to the other system. We loop in this state until the FDIR component, if it behaves according to the specification, switches off the running system. In the next state we (state 1) do not set any input and wait for the FDIR component to eventually switch on one of the systems. If the component switches on the same system, then we go back to the previous state (state 0), if it switches on the other system we go into the next state (state 3).  In this state we have the second system running (\redmode) and set again the \normerr flag, i.e., we again raise a non critical error.  We loop in this state until the FDIR component reacts and, if it conforms to the specification, switches off the running system.  Continuing according to the strategy we always raise a non critical error whatever system the FDIR component activates. Eventually the FDIR component has to activate \safemode or violate the specification.  State 2 is only entered when the FDIR violates G5. In this state, it is irrelevant how the test strategy behaves (as long as the assumptions are satisfied) because the specification has already been violated (which is easy to detect during test execution).

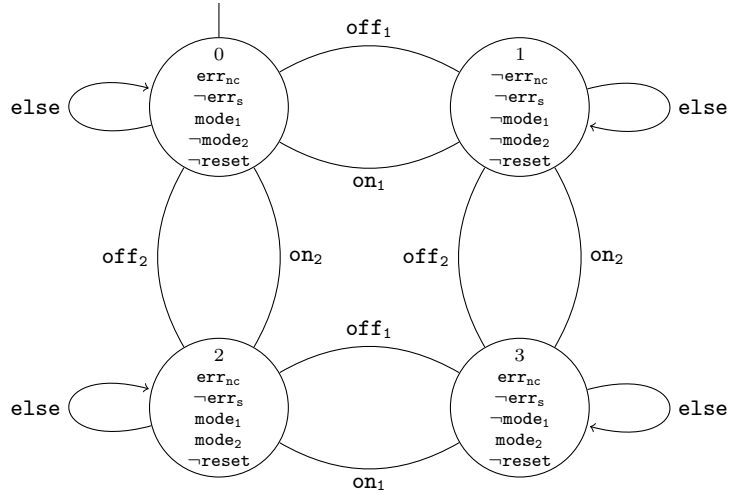
\begin{figure}[t]
  \centering
  \begin{tikzpicture}[auto,node distance=\nd]

\node at (0,5.5) (INIT) {};
  
\node[state,align=center,inner sep=0]  at  (0,4) (S0)
{\begin{minipage}{1cm}\scriptsize\centering 0 \\ \normerr \\       $\neg\criterr$ \\      \nommode \\ $\neg\redmode$ \\ $\neg\reset$\end{minipage}};

\node[state,align=center,inner sep=0]          at  (4,0) (S1)
{\begin{minipage}{1cm}\scriptsize\centering 3 \\ \normerr \\       $\neg\criterr$ \\ $\neg\nommode$ \\      \redmode \\ $\neg\reset$\end{minipage}};
     
\node[state,align=center,inner sep=0]          at  (0,0) (S2)
{\begin{minipage}{1cm}\scriptsize\centering 2 \\ \normerr \\       $\neg\criterr$ \\      \nommode \\       \redmode \\ $\neg\reset$\end{minipage}};
     
\node[state,align=center,inner sep=0]          at  (4,4) (S3)
{\begin{minipage}{1cm}\scriptsize\centering 1 \\ $\neg\normerr$ \\ $\neg\criterr$ \\ $\neg\nommode$ \\ $\neg\redmode$ \\ $\neg\reset$\end{minipage}};

\path
(INIT) edge (S0)
(S0) edge [align=center,loop left]   node[xshift=0mm,yshift=0mm] {{\tt else}} (S0)
(S0) edge [align=center,bend left]   node[xshift=0mm,yshift=0mm] {\redon}     (S2)
(S0) edge [align=center,bend left]   node[xshift=0mm,yshift=0mm] {\nomoff}    (S3)
(S1) edge [align=center,bend left]   node[xshift=0mm,yshift=0mm] {\nomon}     (S2)
(S1) edge [align=center,bend left]   node[xshift=0mm,yshift=0mm] {\redoff}    (S3)
(S1) edge [align=center,loop right]  node[xshift=0mm,yshift=0mm] {{\tt else}} (S1)
(S2) edge [align=center,loop left]  node[xshift=0mm,yshift=0mm]  {{\tt else}} (S2)
(S2) edge [align=center,bend left]   node[xshift=0mm,yshift=0mm] {\nomoff}    (S1)
(S2) edge [align=center,bend left]   node[xshift=0mm,yshift=0mm] {\redoff}    (S0)
(S3) edge [align=center,loop right]  node[xshift=0mm,yshift=0mm] {{\tt else}} (S3)
(S3) edge [align=center,bend left]   node[xshift=0mm,yshift=0mm] {\redon}     (S1)
(S3) edge [align=center,bend left]   node[xshift=0mm,yshift=0mm] {\nomon}     (S0)
;
\end{tikzpicture}
  \caption{Test strategy that tests for a stuck-at-0 fault of signal $\safemode$.}
  \label{fig:FDIRsatmodesafeG0}
\end{figure}

\begin{figure}[t]
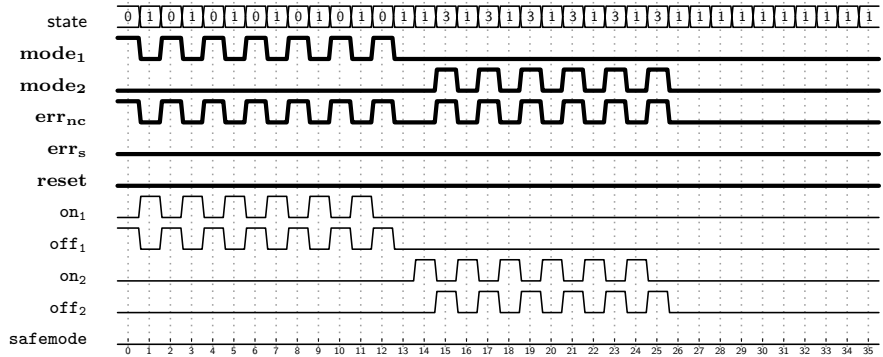

  \centering
  \begin{tikztimingtable}[>=latex,timing/dslope=0.1,timing/.style={x=2ex,y=2ex},x=2ex,timing/rowdist=3ex,timing/name/.style={font=\sffamily\scriptsize},-]
       state & D{0} D{1} D{0} D{1} D{0} D{1} D{0} D{1} D{0} D{1} D{0} D{1} D{0} D{1} D{1} D{3} D{1} D{3} D{1} D{3} D{1} D{3} D{1} D{3} D{1} D{3} D{1} D{1} D{1} D{1} D{1} D{1} D{1} D{1} D{1} D{1}\\
       $\bf mode_1$   & [ultra thick] 1H 1L 1H 1L 1H 1L 1H 1L 1H 1L 1H 1L 1H 17L 1L 5L \\
       $\bf mode_2$   & [ultra thick] 15L 1H 1L 1H 1L 1H 1L 1H 1L 1H 1L 1H 4L 1L 5L \\
       $\bf err_{nc}$  & [ultra thick] 1H 1L 1H 1L 1H 1L 1H 1L 1H 1L 1H 1L 1H 2L 1H 1L 1H 1L 1H 1L 1H 1L 1H 1L 1H 4L 1L 5L \\
       $\bf err_s$    & [ultra thick] 30L 1L 5L \\
       $\bf reset$    & [ultra thick] 30L 1L 5L \\
       $\tt on_1$     & 1L 1H 1L 1H 1L 1H 1L 1H 1L 1H 1L 1H 18L 1L 5L \\
       $\tt off_1$    & 1H 1L 1H 1L 1H 1L 1H 1L 1H 1L 1H 1L 1H 17L  1L 5L \\
       $\tt on_2$     & 14L 1H 1L 1H 1L 1H 1L 1H 1L 1H 1L 1H 5L 1L 5L \\
       $\tt off_2$    & 15L  1H 1L 1H 1L 1H 1L 1H 1L 1H 1L 1H 4L 1L 5L \\
       $\tt safemode$ & 30L 1L 5L \\
       \extracode
       \begin{pgfonlayer}{background}
         \foreach \n in {0,1,...,35}
         \draw (0.5+\n,-18.5+.1) -- +(0,-.2)
         node [below,inner sep=2pt] {\scalebox{.75}{\tiny\n}}[-];
         \begin{scope}[semitransparent ,semithick]
           \vertlines[darkgray,dotted,-]{0.5,1.5 ,...,35.5}
         \end{scope}
       \end{pgfonlayer}
\end{tikztimingtable}
  \caption{Execution trace from a faulty system under the strategy that tests for a stuck-at-0 fault of signal $\safemode$. Bold signals are controlled by the strategy.}
  \label{fig:FDIRExecsatmodesafeStuck0}
\end{figure}

\subsubsection{Test strategy evaluation}
\label{sec:strateval}

\paragraph{Test setting}
In the \eucropis satellite the FDIR component is implemented in C++.  The implementation for the magnetic torquer FDIR handling is not an exact realization of the specification in \reftab{fdir} but extends it by allowing commands to the EP to be lost (e.g. due to electrical faults).  This is accommodated by adding timeouts for the execution of the switch-on/off commands and reissuing the commands if the timeout is triggered.

The implementation is designed with testability and portability in mind and uses an abstract interface to access other sub-systems of the satellite.  This allows to exchange the used interface with a set of test adapters which connect to the signals generated by the test strategies.  As we are only interested in the functional properties of the implementation, we can run the code on a normal Linux system, instead of the microprocessor which is used in the satellite.  This gives access to all Linux based debugging and test tools and allows us to use {\tt gcov} to measure the line and branch coverage of the source code.

A time step of a test run consists of the following operations: request values for the input variables $\inp_{FDIR}$ from the test strategy; feed the values to the test adapter from which they are read by the FDIR implementation; run the FDIR implementation for one cycle; extract the output values $\outp_{FDIR}$ from the test adapter and feed them back to the test strategy to get new input values.  For each time step the execution trace is recorded, i.e., the values assigned to the inputs $\inp_{FDIR}$ and outputs $\outp_{FDIR}$ of the FDIR component.

\paragraph{Mutation testing}
We apply mutation analysis to assess the effectiveness, i.e., fault finding abilities, of a test suite.  A test suite \emph{kills} a mutant program $M$ if it contains at least one test strategy that, when executed on $M$ and the original program $P$, produces a trace where at least one output of $M$ differs in at least one time step from the respective output of $P$ (for the same input sequence).  A mutant program $M$ is \emph{equivalent} to the original program $P$ if $M$ does not violate the specification. For our evaluation we manually identify and remove equivalent mutants.

We generate mutant programs of the C++ implementation of the FDIR component by systematically introducing the following four mutations in each line: 1)~deletion of the line, 2)~replacement of {\tt true} with {\tt false} or {\tt false} with {\tt true}, 3)~replacement of {\tt ==} with {\tt !=} or {\tt !=} with {\tt ==}, and 4)~replacement of {\tt \&\&} with {\tt ||} or {\tt ||} with {\tt \&\&}.  In total, $198$~mutant programs are generated.  We use the GNU compiler {\tt gcc} to remove all mutant programs which do not compile and thus not conform to the C++ programming language.  Also all mutant programs which fail during runtime e.g. by raising a segmentation fault are removed.  We analyzed the remaining $96$ mutants manually and identified $23$ mutants that are correct with respect to the specification, i.e., equivalent mutants. Thus, $73$ mutants violate the specification. Moreover, $11$ of these $73$ mutants can only violate the specification if the $\nomoff$ and $\redoff$ commands can fail, which contradicts our assumptions on the EP unit. We keep those mutants to check whether the strategies can kill them nevertheless. Next, we executed all test strategies on the mutant programs for $80$ time steps each and log the corresponding execution traces.

From the $73$ mutants that violate the specification, our strategies all together are able to kill $52$, i.e., we achieve a mutation score of 71.23\%.  If we do not take the $11$ mutants into account that violate our assumptions for the test strategy generation, then the mutation score increases to 80.65\%.  We illustrate in \reffig{FDIRExecsatmodesafeStuck0} the execution of the test strategy from \reffig{FDIRsatmodesafeG0} on a mutant. This strategy aims for revealing a stuck-at-0 fault of signal $\safemode$. The test strategy first forces the FDIR component to eventually switch to the backup system. The switch happens in time step 14 after several restarts of the system. Then the strategy forces the FDIR component to eventually activate $\safemode$. However, this mutant is faulty and instead of activating $\safemode$ the system remains silent from time step 26 onwards. Thus, violating guarantee G3\footnote{Given that the user has decided that we have waited long enough for $\safemode$ to become true.}.

\begin{table}[t]
  \caption{Mutation coverage by fault models and signal when executing all four derived strategies.}
  \label{tab:fdir:mutcov}
  \centering
  \begin{tabularx}{\textwidth}{p{3cm} ZZZ Z}
    \toprule
    \multicolumn{1}{c}{\bf Output} &
    \multicolumn{4}{c}{\bf Fault Model} \\
    \cmidrule(lr){2-5}
    &
    \multicolumn{1}{c}{\bf S-a-0} &
    \multicolumn{1}{c}{\bf S-a-1} &
    \multicolumn{1}{c}{\bf Bit-Flip} &
    \multicolumn{1}{c}{\bf All} \\
    &
    \multicolumn{1}{c}{[\%]} &
    \multicolumn{1}{c}{[\%]} &
    \multicolumn{1}{c}{[\%]} &
    \multicolumn{1}{c}{[\%]} \\
    \midrule
    \nomon       & 65.75 &  39.73 & 5.48 & 65.75 \\
    \nomoff      &  5.48 &   4.11 & 9.59 &  9.59 \\
    \safemode    & 61.64 &   6.85 & 6.85 & 61.64 \\
    \midrule
    {\bf All}    & 71.23  & 39.73 & 9.59 & 71.23 \\
    \bottomrule
  \end{tabularx}
\end{table}

As the are only derived from requirements,  without any implementation-specific knowledge, they are applicable on any system that claims to implement the given specification. The mutation score of $71.23\%$  illustrates that our strategies, although computed for only three different faults that are assumed to only affect a single output signal, are also sensitive to many other faults.

If we only apply one of the four strategies we computed per fault model and output signal, then the resulting test suite can kill (1) $51$ mutants, (2) $51$ mutants, (3) $49$ mutants and (4) $49$ mutants. While one strategy per fault and output already achieves a high mutation score, these numbers illustrate the advantage of computing multiple strategies per fault model and output signal.

In \reftab{fdir:mutcov} we present the mutation score of the individual combinations of signals and fault models.
From all the mutants killed, there were 9 mutants only killed by a single signal / fault model combination, namely \nomon with stuck-at-0 assumption exclusively killing 7 mutants and \safemode with stuck-at-0 assumption exclusively killing 2 mutants.

\paragraph{Random testing}
We compared the fault finding abilities of the generated test strategies and random testing executed for 100, 10'000, and 100'000 time steps, respectively.  For random testing we use a similar test setup to the test strategy setup, but instead of requesting the input values $\inp_{FDIR}$ from a test strategy we use uniformly distributed random values.  For each time step, the input and output values are recorded.  For each mutant the same input sequence is supplied and the output sequence of the mutant is compared to the output sequence of the actual implementation.  

Random testing for 100 time steps killed 46 mutants (mutation score of 63\%), while random testing for 10'000 time steps killed 69 mutants (mutation score of 94.5\%).  With increased time steps the results stayed the same.  Random testing for 100'000 time steps killed 69 mutants as well.

Our strategies are able to kill three mutants that are missed by all of the three random test sequences.  These mutants can only be killed when executing certain input/output sequences and it is very unlikely for random testing to hit one of the required sequences.  The corresponding sequence requires that a sequence of \normerr, \nommode going low and \nommode going high is executed multiple times before either \criterr or \reset is triggered.

One mutant is neither covered by the test strategies nor by the random sequences.  This mutant requires a longer sequence as well in order to be executed.  The mutant is not covered by the test strategies because the sequence is about the timeout of an EP command, which is not covered by the specification from which the test strategies are derived.

\paragraph{Code coverage}
\reftab{code coverage} lists the line coverage and branch coverage measured with {\tt gcov} for the different testing approaches.  The table is built as follows: each line belongs to one testing approach.  The first column names the approach, the second column lists the number of time steps, and the third and the fourth column present the line and branch coverage.  Overall, the random testing approaches achieve a higher code coverage than the generated adaptive test strategies when executed on the source code of the FDIR component.  The test strategies are directly derived from the specification and independent from a concrete implementation.  Parts of the implementation which refine the specification or which are not specified at all are not necessarily covered.  As mentioned in Section~\ref{sec:strateval} the implementation adds timeouts for operations of the EP.  Manual analysis revealed that removing the corresponding instructions would increase the line coverage to 87.3\% and the branch coverage to 74.5\%.  In combination random tests and our strategies together achieve a line coverage of 97.6\% and a branch coverage of 87\%.

\begin{table}[t]
  \centering
  \caption{Code coverage by testing approach.  The suffix ``k'' multiplies by $10^3$.}
  \label{tab:code coverage}
  \begin{tabularx}{\textwidth}{lZZZ}
    \toprule
    \multicolumn{1}{c}{\bf Approach} &
    \multicolumn{1}{c}{\bf \#Steps} &
    \multicolumn{2}{c}{\bf Coverage Criterion} \\
    \cmidrule(lr){3-4}
    &&
    \multicolumn{1}{c}{\bf Line} &
    \multicolumn{1}{c}{\bf Branch} \\
    &&
    \multicolumn{1}{c}{[\%]} &
    \multicolumn{1}{c}{[\%]} \\
    \midrule
    Random         &     100 & 80.5 & 64.8 \\
    Random         &     10k & 96.3 & 85.2 \\
    Random         &    100k & 96.3 & 85.2 \\
    \midrule
    Test strategy  &      80 & 76.8 & 64.8 \\
    \midrule
    {\bf Together} &         & 97.6 & 87.0 \\
    \bottomrule
  \end{tabularx}
\end{table}

\section{Conclusion}
\label{sec:concl}

We presented a new approach to compute adaptive test strategies from temporal logic specifications using reactive synthesis with partial information. The computed test strategies reveal all instances of a user-defined fault class for every realization of a given specification.  Thus, they do not rely on implementation details, which is important for products that are still under development or for standards that will be implemented by multiple vendors.  Our approach is sound but incomplete in general, i.e., may fail to find test strategies even if they exist.  However, for many interesting cases, we showed that it is both sound and complete.

The worst-case complexity is doubly exponential in the specification size, but in our setting, the specifications are typically small.  This also makes our approach an interesting application for reactive synthesis.  Our experiments demonstrate that our approach can compute meaningful tests for specifications of industrial size and that the computed strategies are capable of detecting faults hidden in paths that are unlikely to be activated by random input sequences.

We applied our approach in a case study on the fault detection, isolation and recovery component of the satellite Eu:CROPIS that is currently under development.  Our computed test suite, based only on three different types of faults, increases the mutation score of random testing from 94.5\% to 98.6\%. We can also increase the branch coverage of the code from 85.2\% to 87\%. In particular, our approach detects faults that require more complex input sequences to be triggered that are not covered by random testing.

Current directions for future work include improving scalability, success-rate, and usability of our approach.  To this end, we are investigating using random testing for inputs in the strategies that are not fixed to single values, and best-effort strategies~\cite{Faella08,Faella09} for the case that there are no test strategies that can guarantee triggering the fault.  Another direction for future work is research on evaluating LTL properties specified on infinite paths on finite traces to improve the evaluation process when executing the derived strategies.

\section*{Acknowledgment}

This work was supported in part by the Austrian Science Fund (FWF) through the research network RiSE (S11406-N23) and by the European Commission through projects IMMORTAL (317753) and eDAS (608770).  We thank Ayrat Khalimov for helpful comments and assistance in using \party.


\begin{thebibliography}{10}

\bibitem{AcreeBDLS79}
Allen~Troy Acree, Timothy~Alan Budd, Richard~A. DeMillo, Richard~J. Lipton, and
  Frederick~Gerald Sayward.
\newblock Mutation analysis.
\newblock Technical Report GIT-ICS-79/08, Georgia Institute of Technology,
  Atlanta, Georgia, 1979.

\bibitem{AichernigBJKST15}
Bernhard~K. Aichernig, Harald Brandl, Elisabeth J{\"{o}}bstl, Willibald Krenn,
  Rupert Schlick, and Stefan Tiran.
\newblock Killing strategies for model-based mutation testing.
\newblock {\em Softw. Test., Verif. Reliab.}, 25(8):716--748, 2015.

\bibitem{AlurCY95}
Rajeev Alur, Costas Courcoubetis, and Mihalis Yannakakis.
\newblock Distinguishing tests for nondeterministic and probabilistic machines.
\newblock In Frank~Thomson Leighton and Allan Borodin, editors, {\em
  Proceedings of the Twenty-Seventh Annual {ACM} Symposium on Theory of
  Computing, 29 May-1 June 1995, Las Vegas, Nevada, {USA}}, pages 363--372.
  {ACM}, 1995.

\bibitem{AmmannDX01}
Paul Ammann, Wei Ding, and Daling Xu.
\newblock Using a model checker to test safety properties.
\newblock In {\em 7th International Conference on Engineering of Complex
  Computer Systems {(ICECCS} 2001), 11-13 June 2001, Sk{\"{o}}vde, Sweden},
  pages 212--221. {IEEE} Computer Society, 2001.

\bibitem{ArmoniFFGPTV03}
Roy Armoni, Limor Fix, Alon Flaisher, Orna Grumberg, Nir Piterman, Andreas
  Tiemeyer, and Moshe~Y. Vardi.
\newblock Enhanced vacuity detection in linear temporal logic.
\newblock In Warren A.~Hunt Jr. and Fabio Somenzi, editors, {\em Computer Aided
  Verification, 15th International Conference, {CAV} 2003, Boulder, CO, USA,
  July 8-12, 2003, Proceedings}, volume 2725 of {\em Lecture Notes in Computer
  Science}, pages 368--380. Springer, 2003.

\bibitem{BauerLS11}
Andreas Bauer, Martin Leucker, and Christian Schallhart.
\newblock Runtime verification for {LTL} and {TLTL}.
\newblock {\em {ACM} Trans. Softw. Eng. Methodol.}, 20(4):14:1--14:64, 2011.

\bibitem{BeerBER01}
Ilan Beer, Shoham Ben{-}David, Cindy Eisner, and Yoav Rodeh.
\newblock Efficient detection of vacuity in temporal model checking.
\newblock {\em Formal Methods in System Design}, 18(2):141--163, 2001.

\bibitem{BlassGNV05}
Andreas Blass, Yuri Gurevich, Lev Nachmanson, and Margus Veanes.
\newblock Play to test.
\newblock In Grieskamp and Weise \cite{DBLP:conf/fates/2005}, pages 32--46.

\bibitem{BloemKPR16}
Roderick Bloem, Robert K{\"{o}}nighofer, Ingo Pill, and Franz R{\"{o}}ck.
\newblock Synthesizing adaptive test strategies from temporal logic
  specifications.
\newblock In Ruzica Piskac and Muralidhar Talupur, editors, {\em 2016 Formal
  Methods in Computer-Aided Design, {FMCAD} 2016, Mountain View, CA, USA,
  October 3-6, 2016}, pages 17--24. {IEEE}, 2016.

\bibitem{BorodayPG07}
Sergiy Boroday, Alexandre Petrenko, and Roland Groz.
\newblock Can a model checker generate tests for non-deterministic systems?
\newblock {\em Electr. Notes Theor. Comput. Sci.}, 190(2):3--19, 2007.

\bibitem{ClarkeE81}
Edmund~M. Clarke and E.~Allen Emerson.
\newblock Design and synthesis of synchronization skeletons using
  branching-time temporal logic.
\newblock In Dexter Kozen, editor, {\em Logics of Programs, Workshop, Yorktown
  Heights, New York, USA, May 1981}, volume 131 of {\em Lecture Notes in
  Computer Science}, pages 52--71. Springer, 1981.

\bibitem{DavidLLN08}
Alexandre David, Kim~Guldstrand Larsen, Shuhao Li, and Brian Nielsen.
\newblock A game-theoretic approach to real-time system testing.
\newblock In Donatella Sciuto, editor, {\em Design, Automation and Test in
  Europe, {DATE} 2008, Munich, Germany, March 10-14, 2008}, pages 486--491.
  {ACM}, 2008.

\bibitem{GiacomoMM14}
Giuseppe {De Giacomo}, Riccardo {De Masellis}, and Marco Montali.
\newblock Reasoning on {LTL} on finite traces: Insensitivity to infiniteness.
\newblock In Carla~E. Brodley and Peter Stone, editors, {\em Proceedings of the
  Twenty-Eighth {AAAI} Conference on Artificial Intelligence, July 27 -31,
  2014, Qu{\'{e}}bec City, Qu{\'{e}}bec, Canada.}, pages 1027--1033. {AAAI}
  Press, 2014.

\bibitem{GiacomoV13}
Giuseppe {De Giacomo} and Moshe~Y. Vardi.
\newblock Linear temporal logic and linear dynamic logic on finite traces.
\newblock In Francesca Rossi, editor, {\em {IJCAI} 2013, Proceedings of the
  23rd International Joint Conference on Artificial Intelligence, Beijing,
  China, August 3-9, 2013}, pages 854--860. {IJCAI/AAAI}, 2013.

\bibitem{DeMilloLS78}
Richard~A. DeMillo, Richard~J. Lipton, and Frederick~G. Sayward.
\newblock Hints on test data selection: Help for the practicing programmer.
\newblock {\em {IEEE} Computer}, 11(4):34--41, 1978.

\bibitem{DilligDMA12}
Isil Dillig, Thomas Dillig, Kenneth~L. McMillan, and Alex Aiken.
\newblock Minimum satisfying assignments for {SMT}.
\newblock In P.~Madhusudan and Sanjit~A. Seshia, editors, {\em Computer Aided
  Verification - 24th International Conference, {CAV} 2012, Berkeley, CA, USA,
  July 7-13, 2012 Proceedings}, volume 7358 of {\em Lecture Notes in Computer
  Science}, pages 394--409. Springer, 2012.

\bibitem{Ehlers12}
R{\"{u}}diger Ehlers.
\newblock Symbolic bounded synthesis.
\newblock {\em Formal Methods in System Design}, 40(2):232--262, 2012.

\bibitem{Faella08}
Marco Faella.
\newblock Best-effort strategies for losing states.
\newblock {\em CoRR}, abs/0811.1664, 2008.

\bibitem{Faella09}
Marco Faella.
\newblock Admissible strategies in infinite games over graphs.
\newblock In Rastislav Kr{\'{a}}lovic and Damian Niwinski, editors, {\em
  Mathematical Foundations of Computer Science 2009, 34th International
  Symposium, {MFCS} 2009, Novy Smokovec, High Tatras, Slovakia, August 24-28,
  2009. Proceedings}, volume 5734 of {\em Lecture Notes in Computer Science},
  pages 307--318. Springer, 2009.

\bibitem{FinkbeinerS13}
Bernd Finkbeiner and Sven Schewe.
\newblock Bounded synthesis.
\newblock {\em {STTT}}, 15(5-6):519--539, 2013.

\bibitem{FraserA08}
Gordon Fraser and Paul Ammann.
\newblock Reachability and propagation for {LTL} requirements testing.
\newblock In Hong Zhu, editor, {\em Proceedings of the Eighth International
  Conference on Quality Software, {QSIC} 2008, 12-13 August 2008, Oxford,
  {UK}}, pages 189--198. {IEEE} Computer Society, 2008.

\bibitem{FraserW07}
Gordon Fraser and Franz Wotawa.
\newblock Test-case generation and coverage analysis for nondeterministic
  systems using model-checkers.
\newblock In {\em Proceedings of the Second International Conference on
  Software Engineering Advances {(ICSEA} 2007), August 25-31, 2007, Cap
  Esterel, French Riviera, France}, page~45. {IEEE} Computer Society, 2007.

\bibitem{FraserWA09b}
Gordon Fraser, Franz Wotawa, and Paul Ammann.
\newblock Issues in using model checkers for test case generation.
\newblock {\em Journal of Systems and Software}, 82(9):1403--1418, 2009.

\bibitem{FraserWA09}
Gordon Fraser, Franz Wotawa, and Paul Ammann.
\newblock Testing with model checkers: a survey.
\newblock {\em Softw. Test., Verif. Reliab.}, 19(3):215--261, 2009.

\bibitem{DBLP:conf/fates/2005}
Wolfgang Grieskamp and Carsten Weise, editors.
\newblock {\em Formal Approaches to Software Testing, 5th International
  Workshop, {FATES} 2005, Edinburgh, UK, July 11, 2005, Revised Selected
  Papers}, volume 3997 of {\em Lecture Notes in Computer Science}. Springer,
  2006.

\bibitem{HavelundR01}
Klaus Havelund and Grigore Rosu.
\newblock Monitoring programs using rewriting.
\newblock In {\em 16th {IEEE} International Conference on Automated Software
  Engineering {(ASE} 2001), 26-29 November 2001, Coronado Island, San Diego,
  CA, {USA}}, pages 135--143. {IEEE} Computer Society, 2001.

\bibitem{Hierons06}
Robert~M. Hierons.
\newblock Applying adaptive test cases to nondeterministic implementations.
\newblock {\em Inf. Process. Lett.}, 98(2):56--60, 2006.

\bibitem{JiaH11}
Yue Jia and Mark Harman.
\newblock An analysis and survey of the development of mutation testing.
\newblock {\em {IEEE} Trans. Software Eng.}, 37(5):649--678, 2011.

\bibitem{JinRS04}
HoonSang Jin, Kavita Ravi, and Fabio Somenzi.
\newblock Fate and free will in error traces.
\newblock {\em {STTT}}, 6(2):102--116, 2004.

\bibitem{KhalimovJB13}
Ayrat Khalimov, Swen Jacobs, and Roderick Bloem.
\newblock {PARTY} parameterized synthesis of token rings.
\newblock In Natasha Sharygina and Helmut Veith, editors, {\em Computer Aided
  Verification - 25th International Conference, {CAV} 2013, Saint Petersburg,
  Russia, July 13-19, 2013. Proceedings}, volume 8044 of {\em Lecture Notes in
  Computer Science}, pages 928--933. Springer, 2013.

\bibitem{KonighoferHB13}
Robert K{\"{o}}nighofer, Georg Hofferek, and Roderick Bloem.
\newblock Debugging formal specifications: a practical approach using
  model-based diagnosis and counterstrategies.
\newblock {\em {STTT}}, 15(5-6):563--583, 2013.

\bibitem{KupfermanV00}
Orna Kupferman and Moshe~Y. Vardi.
\newblock {\em Advances in Temporal Logic}, chapter Synthesis with Incomplete
  Informatio, pages 109--127.
\newblock Springer Netherlands, 2000.

\bibitem{KupfermanV03}
Orna Kupferman and Moshe~Y. Vardi.
\newblock Vacuity detection in temporal model checking.
\newblock {\em {STTT}}, 4(2):224--233, 2003.

\bibitem{LuoBP94}
Gang Luo, Gregor von Bochmann, and Alexandre Petrenko.
\newblock Test selection based on communicating nondeterministic finite-state
  machines using a generalized wp-method.
\newblock {\em {IEEE} Trans. Software Eng.}, 20(2):149--162, 1994.

\bibitem{Martin75}
Donald~A. Martin.
\newblock Borel determinacy.
\newblock {\em Annals of Mathematics}, 102(2):363--371, 1975.

\bibitem{Mathur08}
Aditya~P. Mathur.
\newblock {\em Foundations of Software Testing}.
\newblock Addison-Wesley, second edition edition, 2008.

\bibitem{MiyaseK04}
Kohei Miyase and Seiji Kajihara.
\newblock {XID:} don't care identification of test patterns for combinational
  circuits.
\newblock {\em {IEEE} Trans. on {CAD} of Integrated Circuits and Systems},
  23(2):321--326, 2004.

\bibitem{MorgensternGS12}
Andreas Morgenstern, Manuel Gesell, and Klaus Schneider.
\newblock An asymptotically correct finite path semantics for {LTL}.
\newblock In Nikolaj Bj{\o}rner and Andrei Voronkov, editors, {\em Logic for
  Programming, Artificial Intelligence, and Reasoning - 18th International
  Conference, LPAR-18, M{\'{e}}rida, Venezuela, March 11-15, 2012.
  Proceedings}, volume 7180 of {\em Lecture Notes in Computer Science}, pages
  304--319. Springer, 2012.

\bibitem{NachmansonVSTG04}
Lev Nachmanson, Margus Veanes, Wolfram Schulte, Nikolai Tillmann, and Wolfgang
  Grieskamp.
\newblock Optimal strategies for testing nondeterministic systems.
\newblock In George~S. Avrunin and Gregg Rothermel, editors, {\em Proceedings
  of the {ACM/SIGSOFT} International Symposium on Software Testing and
  Analysis, {ISSTA} 2004, Boston, Massachusetts, USA, July 11-14, 2004}, pages
  55--64. {ACM}, 2004.

\bibitem{Offutt92}
A.~Jefferson Offutt.
\newblock Investigations of the software testing coupling effect.
\newblock {\em {ACM} Trans. Softw. Eng. Methodol.}, 1(1):5--20, 1992.

\bibitem{PetrenkoSY12}
Alexandre Petrenko, Adenilso da~Silva~Sim{\~{a}}o, and Nina Yevtushenko.
\newblock Generating checking sequences for nondeterministic finite state
  machines.
\newblock In Giuliano Antoniol, Antonia Bertolino, and Yvan Labiche, editors,
  {\em Fifth {IEEE} International Conference on Software Testing, Verification
  and Validation, {ICST} 2012, Montreal, QC, Canada, April 17-21, 2012}, pages
  310--319. {IEEE} Computer Society, 2012.

\bibitem{PetrenkoS15}
Alexandre Petrenko and Adenilso Sim{\~{a}}o.
\newblock Generalizing the ds-methods for testing non-deterministic fsms.
\newblock {\em Comput. J.}, 58(7):1656--1672, 2015.

\bibitem{PetrenkoY05}
Alexandre Petrenko and Nina Yevtushenko.
\newblock Conformance tests as checking experiments for partial
  nondeterministic {FSM}.
\newblock In Grieskamp and Weise \cite{DBLP:conf/fates/2005}, pages 118--133.

\bibitem{PetrenkoY14}
Alexandre Petrenko and Nina Yevtushenko.
\newblock Adaptive testing of nondeterministic systems with {FSM}.
\newblock In {\em 15th International {IEEE} Symposium on High-Assurance Systems
  Engineering, {HASE} 2014, Miami Beach, FL, USA, January 9-11, 2014}, pages
  224--228. {IEEE} Computer Society, 2014.

\bibitem{Pnueli77}
Amir Pnueli.
\newblock The temporal logic of programs.
\newblock In {\em 18th Annual Symposium on Foundations of Computer Science,
  Providence, Rhode Island, USA, 31 October - 1 November 1977}, pages 46--57.
  {IEEE} Computer Society, 1977.

\bibitem{PnueliR89}
Amir Pnueli and Roni Rosner.
\newblock On the synthesis of a reactive module.
\newblock In {\em Conference Record of the Sixteenth Annual {ACM} Symposium on
  Principles of Programming Languages, Austin, Texas, USA, January 11-13,
  1989}, pages 179--190. {ACM} Press, 1989.

\bibitem{QueilleS82}
Jean{-}Pierre Queille and Joseph Sifakis.
\newblock Specification and verification of concurrent systems in {CESAR}.
\newblock In Mariangiola Dezani{-}Ciancaglini and Ugo Montanari, editors, {\em
  International Symposium on Programming, 5th Colloquium, Torino, Italy, April
  6-8, 1982, Proceedings}, volume 137 of {\em Lecture Notes in Computer
  Science}, pages 337--351. Springer, 1982.

\bibitem{TanSL04}
Li~Tan, Oleg Sokolsky, and Insup Lee.
\newblock Specification-based testing with linear temporal logic.
\newblock In Du~Zhang, {\'{E}}ric Gr{\'{e}}goire, and Doug DeGroot, editors,
  {\em Proceedings of the 2004 {IEEE} International Conference on Information
  Reuse and Integration, {IRI} - 2004, November 8-10, 2004, Las Vegas Hilton,
  Las Vegas, NV, {USA}}, pages 493--498. {IEEE} Systems, Man, and Cybernetics
  Society, 2004.

\bibitem{TipaldiB15}
Massimo Tipaldi and Bernhard Bruenjes.
\newblock Survey on fault detection, isolation, and recovery strategies in the
  space domain.
\newblock {\em J. Aerospace Inf. Sys.}, 12(2):235--256, 2015.

\bibitem{Yannakakis04}
Mihalis Yannakakis.
\newblock Testing, optimizaton, and games.
\newblock In Josep D{\'{\i}}az, Juhani Karhum{\"{a}}ki, Arto Lepist{\"{o}}, and
  Donald Sannella, editors, {\em Automata, Languages and Programming: 31st
  International Colloquium, {ICALP} 2004, Turku, Finland, July 12-16, 2004.
  Proceedings}, volume 3142 of {\em Lecture Notes in Computer Science}, pages
  28--45. Springer, 2004.

\end{thebibliography}
\end{document}